\documentclass[a4paper,english,10pt]{amsart} 
\usepackage{amsmath}
\usepackage{amsthm}
\usepackage{amsfonts}
\usepackage{amssymb}
\usepackage{graphicx}
\usepackage{xcolor}
\usepackage{enumerate}
\usepackage[english]{babel}
\usepackage{tikz}
\usepackage{fancyhdr}

\theoremstyle{definition}
\newtheorem{definition}{Definition}[section]

\theoremstyle{remark}
\newtheorem{oss}[definition]{Remark}

\theoremstyle{plain}
\newtheorem{theorem}[definition]{Theorem}

\newtheorem{prop}[definition]{Proposition}
\newtheorem{lemma}[definition]{Lemma}

\DeclareMathOperator{\Sym}{Sym}

\DeclareMathOperator{\Pol}{Pol}
\DeclareMathOperator{\Ker}{Ker}
\DeclareMathOperator{\Imag}{Im}

\title[The BV formalism: application to a matrix model]{The BV formalism: theory and application to a matrix model}

\author{Roberta A. Iseppi}
\address{Max Planck Institute for Mathematics, Vivatsgasse 7, 53111 Bonn, Germany}
\email{iseppi@mpim-bonn.mpg.de}

\date{\today}

\begin{document}

\begin{abstract}
\noindent
We review the BV formalism in the context of $0$-dimensional gauge theories. For a gauge theory $(X_{0}, S_{0})$ with an affine configuration space $X_{0}$, we describe an algorithm to construct a corresponding extended theory $(\widetilde{X}, \widetilde{S})$, obtained by introducing ghost and anti-ghost fields, with $\widetilde{S}$ a solution of the classical master equation in $\mathcal{O}_{\widetilde{X}}$. This construction is the first step to define the (gauge-fixed) BRST cohomology complex associated to $(\widetilde{X}, \widetilde{S})$, which encodes many interesting information on the initial gauge theory $(X_{0}, S_{0})$. The second part of this article is devoted to the application of this method to a matrix model endowed with a $U(2)$-gauge symmetry, explicitly determining the corresponding $\widetilde{X}$ and the general solution $\widetilde{S}$ of the classical master equation for the model.
\end{abstract}

\maketitle

\section{Origin and motivation for the BV formalism: a short introduction}
\label{Introduction}
\noindent
Quantization of a gauge field theory is notoriously difficult, even if one takes a formal approach and works with Feynman's path integral \cite{Feynman}. The problem is caused by the very nature of a gauge theory: indeed, while a physical theory can be described as a pair {\small{$(X_{0}, S_{0})$}}, with {\small{$X_{0}$}} the {\em (field) configuration space} and {\small{$S_{0}:X_{0} \rightarrow \mathbb{R}$}} the {\em action functional}, a gauge theory has to be thought as a triple where, next to the configuration space {\small{$X_{0}$}} and the action {\small{$S_{0}$}}, there is also a so-called {\em gauge group} {\small{$\mathcal{G}$}} which describes the symmetries of the theory (cf. Definition \ref{Def: gauge theory}). It is precisely the presence of local symmetries in the action functional {\small{$S_{0}$}} which impedes the direct application of the canonical quantization via path integral of this class of theories. In 1967 Faddeev and Popov  (cf. \cite{Faddeev-Popov}) gave a fundamental contribution to the solution of this problem: their main idea was to introduce, next to the physical fields already in the theory, other {\em non-physical fields} (called {\em Faddeev-Popov ghost fields}), used to eliminate the local symmetries appearing in the path integral and avoid the degeneracy of the propagator. In the years that followed, this idea has been developed (cf. \cite{Zinn-Justin}), arriving to the formulation of the BV construction, discovered by Batalin and Vilkovisky (cf. \cite{BV1}, \cite{BV2}). \\
\\
Thus, the BV formalism found its origin and motivation in the context of the quantization of gauge theory. However, this formalism has shown to be a very interesting construction on its own and not only as a tool in the BRST quantization procedure. For this reason, recently the BV formalism has been thoroughly investigated for infinite-dimensional gauge theories (review articles \cite{AKSZ}, \cite{BBH}, \cite{BBH2}, \cite{Fuster}, \cite{GPS}, \cite{Henneaux2}) as well as in the finite-dimensional context (cf. \cite{Fior}, \cite{Schw}). \\
\\
The purpose of this article is to present the analysis of the BV formalism in the context of gauge theories defined over a $0$-dimensional spacetime. This has demonstrated to be a surprisingly rich context for the analysis of the BV construction from a completely algebraic point of view. The BV formalism can then be simply viewed as a method to construct, given an initial gauge theory {\small{$(X_{0}, S_{0})$}}, with {\small{$X_{0}$}} the {\em initial configuration space} and {\small{$S_{0}$}} the {\em initial action}, a new pair {\small{$(\widetilde{X}, \widetilde{S})$}}, where the {\em extended configuration space} {\small{$\widetilde{X}$}} is obtained as an extension of {\small{$X_{0}$}} with {\em ghost/anti-ghost fields}:
$$\widetilde{X} = X_{0} \cup \{ \mbox{ghost/anti-ghost fields} \},$$
and the {\em extended action} {\small{$\widetilde{S}$}} is defined by adding extra terms depending on the ghost/anti-ghost fields to the initial action {\small{$S_{0}$}}:
$$\widetilde{S} = S_{0} + \mbox{terms depending on ghost/anti-ghost fields}.$$
In order to properly extend the gauge theory {\small{$(X_{0}, S_{0})$}}, further requirements need to be imposed on the new pair {\small{$(\widetilde{X}, \widetilde{S})$}} (cf. Section \ref{Section: The BV formalism}). As a consequence, each extended pair {\small{$(\widetilde{X}, \widetilde{S})$}} naturally induces a cohomology complex, known as the {\em classical BRST cohomology complex}. The discovery that the ghost/anti-ghost fields can be viewed as generators of a cohomology complex should be ascribed to Becchi, Rouet, Stora (cf. \cite{BRS}, \cite{BRS3}) and, independently, Tyutin (cf. \cite{T}), in 1975, after whom this cohomology complex is named. Moreover, under certain conditions, this cohomology complex is still present also after the {\em gauge-fixing procedure} has been implemented to remove the anti-fields/anti-ghost fields both from the extended configuration space {\small{$\widetilde{X}$}} and the extended action {\small{$\widetilde{S}$}}. This {\em gauge-fixed BRST cohomology complex} plays an interesting role also from a physical point of view: indeed, 
the cohomology induced by the gauge-fixed extended theory {\small{$(\widetilde{X}, \widetilde{S})|_{\Psi}$}}, with $\Psi$ a suitable gauge-fixing fermion, allows to recover physically relevant information on the initial gauge theory {\small{$(X_{0}, S_{0})$}} such as the gauge-invariant functions, i.e. the (classical) observables of {\small{$(X_{0}, S_{0})$}}:
$$H^{0}(\widetilde{X}|_{\Psi}, d_{\widetilde{S}}|_{\Psi}) = \{ \mbox{ Observables of the initial gauge theory } (X_{0}, S_{0})\ \}.$$
In this article, we will concentrate on this first part of the construction, namely on the BV formalism, analyzing it from a totally algebraic point of view, postponing the second part of the construction, that is, the presentation of the gauge-fixing procedure and the explicit description of the gauge-fixed BRST cohomology complex to \cite{articolo_cohomology}. \\
\\
We focus on $0$-dimen\-sio\-nal gauge theories, that is, theories where the configuration space {\small{$X_{0}$}} is finite dimensional. The article is structured as follows: the notion of {\em extended theory} is presented in Section \ref{Section: The BV formalism} as the mathematical object to describe the pair {\small{$(\widetilde{X}, \widetilde{S})$}} obtained as the extension of a theory {\small{$(X_{0}, S_{0})$}} through the introduction of ghost/anti-ghost fields. Section \ref{Section: Construction of extended varieties} is devoted to explain a new procedure, inspired by the method in \cite{felder}, to construct a pair {\small{$(\widetilde{X}, \widetilde{S})$}} given a $0$-dimensional gauge theory {\small{$(X_{0}, S_{0})$}}. This method suggests a possible way to face a problem that often appears when the BV construction is applied in the finite-dimensional context: the emergence of an {\em infinite} number of ghost/anti-ghost fields to be added to the theory. The novel construction we present gives a method to select a {\em finite} family of anti-ghost fields (and so of ghost fields). Furthermore, these anti-ghost fields can be algorithmically determined, avoiding the full description of a Tate resolution, which may impose a heavy computation. This BV construction is applied in Section \ref{Section: Extended varieties for a matrix model} to a matrix model with a {\small{$U(2)$}}-gauge symmetry. The implementation of this procedure on an explicit model is relevant because not only it allows to understand the meaning of this {\em finite} family of ghost fields and to determine the characteristics of the gauge theory that it detects, but also because, determining a (generic) solution {\small{$\widetilde{S}$}} for the classical master equation in {\small{$\mathcal{O}_{\widetilde{X}}$}}, it provides all the necessary elements to explicitly determining generators and coboundary operator for the classical BRST complex and hence for the gauge-fixed BRST complex (for this second part of the construction on the {\small{$U(2)$}}-matrix model we refer to the forthcoming \cite{articolo_cohomology}).\\
\\
\noindent
{\em Acknowledgement:} the research presented in this article was supported by Netherlands Organization for Scientific Research (NWO), through Vrije Competitie (pro\-ject number 613.000.910). The author would also like to thank the Max Planck Institute for Mathematics in Bonn for hospitality and support during the final stages of writing this article. Moreover, many thanks are owed to Giovanni Felder for his accurate remarks. Finally, special thanks are due to the author's supervisor Walter D. van Suijlekom, for having introduced her to the subject, for many helpful discussions and for carefully reading the drafts of this article. 
 
\section[The BV formalism]{The BV formalism}
\label{Section: The BV formalism}
\noindent
Given an initial gauge theory {\small{$(X_{0}, S_{0})$}}, with {\small{$X_{0}$}} the initial configuration space and {\small{$S_{0}$}} the initial action functional on {\small{$X_{0}$}}, we want to describe a method to extend it to a new pair {\small{$(\widetilde{X}, \widetilde{S})$}}:
$$\begin{array}{ccc}
(X_{0}, S_{0}) & -------\rightarrow  &(\widetilde{X}, \widetilde{S})\\
\mbox{\small{initial gauge theory}} & \mbox{\tiny{BV construction}} & \mbox{\small{extended theory}}
\end{array}
$$
 In this section we recall the requirements on the pair {\small{$(\widetilde{X}, \widetilde{S})$}} that are imposed from a physical point of view and we introduce the mathematical notion that describes this extended pair, that is, the notion of {\em extended theory}. To describe the BV construction we restrict to consider $0$-dimensional gauge theories, whose definition we recall for completeness. 

\begin{definition}
\label{Def: gauge theory}
Let {\small{$X_{0}$}} be a vector space over {\small{$\mathbb{R}$}}, {\small{$S_{0}$}} be a functional on {\small{$X_{0}$}}, {\small{$S_{0}: X_{0} \rightarrow \mathbb{R}$}}, and {\small{$\mathcal{G}$}} be a group acting on {\small{$X_{0}$}} through an action {\small{$F:\mathcal{G} \times X_{0} \rightarrow X_{0}.$}} Then the pair {\small{$(X_{0}, S_{0})$}} is a {\em gauge theory with gauge group {\small{$\mathcal{G}$}}} if it holds that
 $$S_{0}(F(g, \varphi)) = S_{0}(\varphi), \quad \quad \forall \varphi \in X_{0}, \ \forall g \in \mathcal{G}.$$
\end{definition}

Concerning the terminology, {\small{$X_{0}$}} is called the {\em configuration space}, an element $\varphi$ in {\small{$X_{0}$}} is a {\em gauge field}, the functional {\small{$S_{0}$}} is the {\em action}, and {\small{$\mathcal{G}$}} is known as the {\em gauge group}.

\begin{definition}
A {\em field/ghost field} $\varphi$ is a graded variable characterized by two integers:
 $$\deg(\varphi) \in \mathbb{Z} \quad \mbox{ and } \quad \epsilon(\varphi) \in \{ 0, 1 \},  \quad \mbox{ with } \quad \deg(\varphi) = \epsilon(\varphi) \quad (\mbox{mod} \ \mathbb{Z}/2).$$
$\deg(\varphi)$ is the {\em ghost degree}, while $\epsilon(\varphi)$ is the {\em parity}, which distinguishes between the bosonic case, where $\epsilon(\varphi)=0$ and $\varphi$ behaves as a real variable, and the fermionic case, where $\epsilon(\varphi)=1$ and $\varphi$ behaves as a Grassmannian variable: 
$$\varphi \psi = - \psi \varphi, \quad \quad \mbox{ and }\quad\quad \varphi^{2} =0, \quad \quad  \mbox { if } \quad \epsilon(\varphi) = \epsilon(\psi) = 1.$$
The {\em anti-field/anti-ghost field} $\varphi^{*}$ corresponding to a field/ghost field $\varphi$ satisfies
$$\deg(\varphi^{*}) = - \deg(\varphi) -1, \quad \quad \mbox{ and } \quad \quad \epsilon(\varphi^{*}) = \epsilon(\varphi) +1, \quad (\mbox{mod} \ \mathbb{Z}/2).$$
\end{definition}
In what follows, the term {\em fields} is reserved to the initial fields in {\small{$X_{0}$}} while {\em ghost fields} is used to identify the extra fields introduced by the BV construction. Analogously, {\em anti-fields} is specifically used for the anti-particles corresponding to the initial fields while the {\em anti-ghost fields} are the ones corresponding to the ghost fields.\\
\\
In the BV formalism, the {\em extended configuration space} {\small{$\widetilde{X}$}} is required to be a 
super graded vector space, suitable to be decomposed as the following direct sum:
\begin{equation}
\label{eq: decomp X-tilde}
\widetilde{X} \cong \mathcal{F} \oplus \mathcal{F}^{*}[1], \quad \quad \mbox{ with } \quad [\widetilde{X}]^{0} = X_{0}
\end{equation}
for {\small{$\mathcal{F} = \oplus_{i \geqslant 0} \mathcal{F}^{i}$}} a graded locally free {\small{$\mathcal{O}_{X_{0}}$}}-module with homogeneous components of finite rank and with {\small{$\mathcal{O}_{X_{0}}$}} the algebra of regular functions on {\small{$X_{0}$}}. While $\mathcal{F}$ describes the fields/ghost-fields content of {\small{$\widetilde{X}$}}, {\small{$\mathcal{F}^{*}[1]$}} determines the anti-fields/anti-ghost fields part, with {\small{$\mathcal{F}^{*}[1]$}} that denotes the shifted dual module of $\mathcal{F}$:
 $$\mathcal{F}^{*}[1] = \oplus_{i \in \mathbb{Z}} \big[ \mathcal{F}^{*}[1] \big]^{i} \quad \quad \mbox{ with } \quad \big[ \mathcal{F}^{*}[1] \big]^{i} = \big[ \mathcal{F}^{*} \big]^{i+1}.$$
The condition of {\small{$\widetilde{X}$}} being a {\em super} graded vector space encodes the fact that even-graded elements in {\small{$\widetilde{X}$}} behaves as real variables while odd-graded elements are treated as Grassmannian variables.

\begin{oss}
 Each homogeneous component of the graded vector space {\small{$\widetilde{X}$}} is supposed to be finite-dimensional. However, no hypothesis is done on the number of non-trivial homogeneous components in $\widetilde{X}$, which may be infinite. Therefore, for each ghost degree there is a finite number of independent ghost/anti-ghost fields with that degree while it is allowed the possibility of having independent ghost/anti-ghost fields with ghost degree any integer value. 
\end{oss}

Given a super graded vector space {\small{$\widetilde{X}$}}, by {\small{$\mathcal{O}_{\widetilde{X}}$}} we denote the {\em algebra of regular functions} on {\small{$\widetilde{X}$}}, which is the graded symmetric algebra generated by {\small{$\widetilde{X}$}} over the ring {\small{$\mathcal{O}_{X_{0}}$}}:
$$\mathcal{O}_{\widetilde{X}} = Sym_{\mathcal{O}_{X_{0}}}(\mathcal{F} \oplus \mathcal{F}^{*}[1]).$$  
Due to the graded structure on {\small{$\widetilde{X}$, $\mathcal{O}_{\widetilde{X}}$}} is naturally endowed with a graded algebra structure. Moreover, {\small{$\mathcal{O}_{\widetilde{X}}$}} can be equipped with a graded Poisson bracket structure of degree $1$:
 $$\left\lbrace -, - \right\rbrace: [\mathcal{O}_{\widetilde{X}}]^{i}\otimes [\mathcal{O}_{\widetilde{X}}]^{j}\rightarrow [\mathcal{O}_{\widetilde{X}}]^{i+j+1},$$
which is completely determined by imposing that it acts as follows on the generators:
$$\big\{ \beta_{i}, \beta_{j}\big\}= 0 , \quad \quad \quad \big\{ \beta^{*}_{i}, \beta_{j}\big\} = \delta_{ij} \quad  \quad \mbox{ and } \quad \quad \big\{ \beta^{*}_{i}, \beta^{*}_{j}\big\}=0 $$  
for {\small{$\beta_{i} \in \mathcal{F}^{p}$}} and {\small{$\beta^{*}_{i} \in \big[ \mathcal{F}^{*}[1] \big]^{-p-1}$}}, and then extending it by requiring being linear and graded Poisson. More explicitly, we are first of all imposing that the bracket is {\em graded symmetric}, that is, it satisfies 
$$\left\lbrace a, b \right\rbrace = - (-1)^{(|a| -1) (|b|-1)} \left\lbrace b, a \right\rbrace,$$
for $a, b, c \in$ {\small{$\mathcal{O}_{\widetilde{X}}$}}, where $|a|$ denotes the degree of an element in {\small{$\mathcal{O}_{\widetilde{X}}$}}. Then, we demand that it is also {\em graded Poisson}, so that it holds the following equality:
$$\left\lbrace ab, c \right\rbrace =a \left\lbrace b,c \right\rbrace + (-1)^{|a||b|}b \left\lbrace a,c \right\rbrace,$$
where, once again, $a, b, c$ are elements in {\small{$\mathcal{O}_{\widetilde{X}}$}}. Finally, the bracket has to satisfy also the {\em graded Jacobi identity}, which reads as follows:
{\small{$$(-1)^{(|a|-1)(|c|-1)}\left\lbrace a, \left\lbrace b,c \right\rbrace \right\rbrace+ (-1)^{(|b|-1)(|a|-1)}\left\lbrace b, \left\lbrace c,a \right\rbrace\right\rbrace + (-1)^{(|c|-1)(|b|-1)}\left\lbrace c, \left\lbrace a,b \right\rbrace\right\rbrace=0.$$}}
This Poisson bracket structure on {\small{$\mathcal{O}_{\widetilde{X}}$}} enters the notion of {\em extended action} {\small{$\widetilde{S}$}}, as recalled in the following definition.

\begin{definition}
Let the pair {\small{$(X_{0}, S_{0})$}} be a gauge theory. Then an {\em extended theory} associated to {\small{$(X_{0}, S_{0})$}} is a pair {\small{$(\widetilde{X}, \widetilde{S})$}} where {\small{$\widetilde{X} = \oplus_{i \in \mathbb{Z}} [\widetilde{X}]^{i}$}} is a super graded vector space as in \eqref{eq: decomp X-tilde} and {\small{$\widetilde{S} \in [\mathcal{O}_{\widetilde{X}}]^{0}$}} is a regular function on {\small{$\widetilde{X}$}}, with {\small{$\widetilde{S}|_{X_{0}}=S_{0}$}}, {\small{$\widetilde{S}\neq S_{0}$}} and such that it solves the {\em classical master equation}, i.e., 
$$\{\widetilde{S}, \widetilde{S}\}=0,$$
where {\small{$\{ -, -\}$}} denotes the graded Poisson structure on the algebra {\small{$\mathcal{O}_{\widetilde{X}}$}}.
 \end{definition}

Even though this notion of extended theory is similar to the notion of {\em  BV variety}, introduced by Felder and Kazhdan in \cite{felder}, two main differences distinguish them. The first difference is technical: instead of allowing the configuration space {\small{$X_{0}$}} to be a {\em nonsingular algebraic variety}, we require it to be an {\em affine space}. This permits to give a global description of the BV construction, that is, we can describe {\small{$\widetilde{X}$}} as a graded vector space  instead of having to define it as a sheaf of modules. However, because the matrix models we are interested in satisfy this stronger requirement of having an affine configuration space {\small{$X_{0}$}}, we prefer to avoid the technical difficulties related to having to work with local descriptions. Concerning the second difference, this is more conceptual: we do not require the negatively graded cohomology complex induced by the pair {\small{$(\widetilde{X}, \widetilde{S})$}} to be acyclic. As more precisely explained in Section \ref{Section: Construction of extended varieties}, by omitting this extra condition we are able to develop a construction where only a {\em finite} number of ghost fields has to be added.
 
 \begin{oss}
A fundamental consequence of {\small{$\widetilde{S}$}} solving the classical master equation is that a differential operator is induced over {\small{$\mathcal{O}_{\widetilde{X}}$}}. Indeed, given {\small{$\widetilde{S} \in [\mathcal{O}_{\widetilde{X}}]^{0}$}} such that {\small{$\{ \widetilde{S}, \widetilde{S}\} =0$}}, the map 
$$ d_{\widetilde{S}}: [\mathcal{O}_{\widetilde{X}}]^{\bullet}  \longrightarrow   [\mathcal{O}_{\widetilde{X}}]^{\bullet+1} \quad \quad \mbox{ with } \quad  d_{\widetilde{S}}(\varphi):=\big\{\widetilde{S}, \varphi \big\} , $$
defines a linear differential operator of degree $1$ over the graded Poisson algebra {\small{$\mathcal{O}_{\widetilde{X}}$}}. While the linearity of {\small{$d_{\widetilde{S}}$}} as well as its being of degree $1$ are immediate consequences of properties of the Poisson bracket, the fact that {\small{$d_{\widetilde{S}}$}} is a differential, i.e., that it satisfies {\small{$d_{\widetilde{S}}^{2}=0$}}, follows from the graded Jacobi identity and {\small{$\widetilde{S}$}} solving the classical master equation. Moreover, due to the properties of the Poisson bracket on {\small{$\mathcal{O}_{\widetilde{X}}$}}, the operator {\small{$d_{\widetilde{S}}$}} is {\em graded derivative}: 
$$d_{\widetilde{S}}(\varphi\psi)=(d_{\widetilde{S}}(\varphi))\psi+ (-1)^{\deg(\varphi)}\varphi d_{\widetilde{S}}(\psi),$$ 
for $\varphi, \psi \in$ {\small{$\mathcal{O}_{\widetilde{X}}$}}, and {\em graded distributive} when composed together with the bracket:
$$d_{\widetilde{S}}(\left\lbrace \varphi, \psi\right\rbrace)= \left\lbrace d_{\widetilde{S}}(\varphi), \psi\right\rbrace + (-1)^{\deg(\varphi)-1} \left\lbrace \varphi, d_{\widetilde{S}}(\psi)\right\rbrace.$$
Hence {\small{$(\mathcal{O}_{\widetilde{X}}, \{ - , - \}, d_{\widetilde{S}})$}} defines a {\em differential $P_{0}$-algebra} (cf. \cite{Costello}).
\end{oss}

The differential operator {\small{$d_{\widetilde{S}}$}} acts as coboundary operator for the {\em classical BRST cohomology complex}.

\begin{definition}
\label{definition classical BRST complex}
Given an extended theory {\small{$(\widetilde{X}, \widetilde{S})$}}, the corresponding {\em classical BRST cohomology complex} is a complex where the cochain spaces are 
$$\mathcal{C}^{i}(\widetilde{X}, d_{\widetilde{S}}) = [Sym_{\mathcal{O}_{X_{0}}}(\widetilde{X})]^{i},  $$
for $ i \in \mathbb{Z}$, and the coboundary operator is {\small{$d_{\widetilde{S}}:= \{ \widetilde{S}, - \}$}}.
\end{definition}

\section{Construction of an extended theory}
\label{Section: Construction of extended varieties}
\noindent
In this section we present a method to explicitly construct, from an initial gauge theory {\small{$(X_{0}, S_{0})$}} a corresponding extended theory {\small{$(\widetilde{X}, \widetilde{S})$}}. This method makes fundamental use of a Tate resolution of the Jacobian ring {\small{$J(S_{0})$}}. A crucial passage is to select a {\em finite} family of generators, among all (often infinitely many) generators introduced in a Tate resolution. This finite family of generators determines a finite family of ghost/anti-ghost fields, used to extend the initial configuration space {\small{$X_{0}$}}. The method we are going to explain is inspired by the construction presented by Felder and Kazhdan in \cite{felder}. However, it is precisely this selection of the generators that distinguishes our method and allows us to have an explicit description of the extended theory {\small{$(\widetilde{X}, \widetilde{S})$}} and of the corresponding BRST cohomology complex. 

\begin{definition}
Given a gauge theory {\small{$(X_{0}, S_{0})$}}, with {\small{$X_{0}$}} an $m$-dimensional affine space, the {\em Jacobian ring} {\small{$J(S_{0})$}} of {\small{$S_{0}$}} is the quotient {\small{$J(S_{0}) = \mathcal{O}_{X_{0}} / \Imag(\delta)$}}, for 
 $$ \delta:  T_{X_{0}} \longrightarrow \mathcal{O}_{X_{0}} \quad \quad \mbox{ with } \quad \quad \delta(\xi) : = \xi(S_{0}). $$
\end{definition}
Equivalently, given a global system of coordinates {\small{$\{ x_{i}\}$, $i=1, \dots, m$}}, on {\small{$X_{0}$}}, the Jacobian ring {\small{$J(S_{0})$}} is suitable for the following more explicit description:  
\begin{equation}
J(S_{0})= \frac{\mathcal{O}_{X_{0}}}{\langle \partial_{1}S_{0}, \cdots, \partial_{m}S_{0} \rangle} \ ,
\end{equation}
for {\small{$\partial_{i}S_{0}$}} the partial derivative of {\small{$S_{0}$}} with respect to the $i-$th coordinate {\small{$x_{i}$}}. As previously announced, a key role in the construction of the pair $(\widetilde{X}, \widetilde{S})$ is played by a {\em Tate resolution} of $J(S_{0})$, whose construction we briefly recall. 

\subsection{Tate's algorithm}
\label{Tate's algorithm}
Since the Tate resolution (cf. \cite{Tate}) is an important tool with a broad range of applications, we review it in a more general context than the one in which we are going to use it. Thus let $R$ be a commutative Noetherian ring with unit element and let $M$ be an ideal in $R$. Tate's algorithm is a canonical procedure for constructing a {\em free resolution} of $R/M$ that is a {\em differential $R$-algebra}.

\begin{definition}
Let {\small{$A=\oplus_{i \in \mathbb{Z}_{\leqslant 0}} A_{i}$}} be a {\small{$\mathbb{Z}_{\leqslant 0}$}}-graded commutative algebra whose homogeneous components {\small{$A_{i}$}} are finitely generated as $R$-modules and such that {\small{$A_{0}\cong R \cdot 1_{A}$}}, for {\small{$1_{A}$}} a unit element. Given {\small{$\delta =\lbrace \delta_{i} \rbrace_{i \in \mathbb{Z}_{\leqslant 0}},$}} with {\small{$\delta_{i}: A_{i}\rightarrow A_{i+1}$}}, a graded derivation of degree $1$ satisfying {\small{$\delta^2 = 0$}}, the pair {\small{$(A, \delta)$}} defines a {\em differential $R$-algebra}. 
\end{definition}

A differential $R$-algebra $(A, \delta)$ can equivalently be viewed as a complex of finitely generated $R$-modules with a coboundary operator $\delta$:
$$
\cdots  \xrightarrow{\delta_{-n-1}} A_{-n}  \xrightarrow{\delta_{-n}} A_{-n+1} \xrightarrow{\delta_{-n+1}}
 \cdots  \xrightarrow{\delta_{-2}} A_{-1} \xrightarrow{\delta_{-1}} A_{0}\cong R  \xrightarrow{\delta_{0}}
 0. 
$$
Therefore, using the standard terminology, the {\em cohomology algebra} $H(A)$ of a differential $R$-algebra $A$ is defined as follows:
$$H(A)= \bigoplus_{k \in \mathbb{Z}_{\leqslant 0}} H^{k}(A) \quad \quad \mbox{with} \quad \quad  H^{k}(A)= \frac{\Ker(\delta_{k})}{\Imag(\delta_{k-1})}.$$

\begin{definition}
A differential $R$-algebra $(A, \delta)$ is a {\em Tate resolution} of the $R$-module $R/M$ if $A$ is {\em free}, that is, if each homogeneous component {\small{$A_{i}$}} is a free $R$-module, and {\em acyclic}, i.e., it holds that:
$$H(A)=H^{0}(A) = R/M.$$
\end{definition}

The conditions stated in the above definition can be equivalently rephrased by requiring that the {\small{$A_{i}$}} are free $R$-modules and the following sequence is exact: 
\begin{equation}
\cdots  \xrightarrow{\delta_{-3}} A_{-2} \xrightarrow{\delta_{-2}} A_{-1}  \xrightarrow{\delta_{-1}} R  \xrightarrow{\pi} R/M \rightarrow 0  \mbox{ } ,
\end{equation}
where $\pi$ is the canonical projection map. \\
\\
Tate's algorithm is an inductive construction to determine, given a ring $R$ and an ideal $M$, a free and acyclic differential $R$-algebra $(A, \delta)$ such that {\small{$H^{0}(A) = R/M$}}. This algebra $A$ is obtained as the union of an ascending chain of differential $R$-algebras {\small{$(A^{i}, \delta^{i})$}}, with
$$
A^{0}:=R\subseteq A^{-1}\subseteq A^{-2}\subseteq \cdots
$$
Thus we begin by describing the first step of this algorithm, that is, the construction of the differential $R$-algebra {\small{$(A^{-1}, \delta^{-1})$}}, and then we explain the inductive step, that is, the construction of the differential $R$-algebra {\small{$(A^{-(k+1)}, \delta^{-(k+1)})$}}, given {\small{$(A^{-k}, \delta^{-k})$}}.\\
\\
\noindent
{\em First step: the construction of }{\small{$(A^{-1}, \delta^{-1})$}}. We want to construct a differential $R$-algebra {\small{$(A^{-1}, \delta^{-1})$}} such that it determines a free and acyclic resolution of $R/M$ up to degree $0$. Let {\small{$\{ \tau_{i} \}_{i}$}}, with $i= 1, \dots, n$, be a finite set of generators for $M$ as an $R$-module. Then {\small{$A^{-1}$}} is defined as the following extension:
$$A^{-1} = \Pol_{R}(T_1, \dots ,T_n),$$
with {\small{$\{T_{i}\}$}}, for $i=1, \dots, n$, a family of formal Grassmannian variables of degree $-1$. The differential {\small{$\delta^{-1}=\{ \delta^{-1}_{j}\}$}}, {\small{$j\in \mathbb{Z}_{\leqslant 0}$}}, on the algebra {\small{$A^{-1}$}} is uniquely determined by imposing that
$$\delta^{-1}_{-1}(T_i)=\tau_i , $$
for $i=1, \dots, n$, and then extending this to a map on the whole {\small{$A^{-1}$}} by R-linearity. Due to this definition of {\small{$\delta^{-1}$}}, it holds that 
$$H^{0}(A^{-1})=R/M$$
and hence the sequence
$$
A^{-1}_{-1} 
 \xrightarrow{\delta^{-1}_{-1}}
A^{0}_{0}= R 
 \xrightarrow{\pi}
 R/M  \rightarrow 
0 \ 
$$
is exact. So {\small{$(A^{-1}, \delta^{-1})$}} gives a free and acyclic resolution of $R/M$ up to degree $0$.\\
\\
\noindent
{\em The inductive step: the construction of {\small{$(A^{-(k+1)}, \delta^{-(k+1)})$}}, for $k>0$}. Let {\small{$(A^{-k}, \delta^{-k})$}} be a differential $R$-algebra that induces a free and acyclic resolution of $R/M$ up to degree {\small{$-k+1$}}. In other words, we are assuming that the following sequence of finitely generated modules over $R$ is exact:
\begin{equation}
A^{-k}_{-k} 
 \xrightarrow{\delta^{-k}_{-k}}
 \dots \xrightarrow{\delta^{-2}_{-2}} A^{-1}_{-1} 
 \xrightarrow{\delta^{-1}_{-1}}
A^{0}_{0}= R 
 \xrightarrow{\pi}
 R/M  \rightarrow 
0 \ ,
\label{xk}
\end{equation}
and that is holds:
$$H^{-j}(A^{-k}) = 
\left\lbrace \begin{array}{ll}
R/M & \quad  \mbox{if} \quad j=0\\
0 & \quad  \mbox{if} \quad j=1, \dots, k-1 \ .
\end{array}
\right.
$$
\noindent
The aim is now to define a differential $R$-algebra {\small{$(A^{-(k+1)}, \delta^{-(k+1)})$}} such that, when adding the module {\small{$A^{-(k+1)}_{-(k+1)}$}} to the sequence (\ref{xk}), we obtain a free and acyclic resolution of $R/M$ up to degree $-k$. To achieve this goal let {\small{$\{\tau_i\}$}}, with {\small{$i=1, \dots, n_{k}$}} be $(-k)$-cocycles such that their corresponding cohomology classes generate {\small{$H^{-k}(A^{-k})$}} as finitely-generated module over $R$. We define the graded $R$-algebra {\small{$A^{-(k+1)}$}} as the extension of the algebra {\small{$A^{-k}$}} by the introduction of a family {\small{$\{ T_{i} \}_{i}$, $i=1, \dots, n_{k}$}}, of formal variables of degree $-(k+1)$. Hence: 
$$A^{-(k+1)} = \Pol_{A^{-k}}(T_1, \dots, T_{n_{k}}),$$
where the variables {\small{$\{ T_{i} \}$}} are {\em real} variables if their degree $-(k+1)$ is even or {\em Grassmannian} if the degree $-(k+1)$ is odd. Finally, the derivation {\small{$\delta^{-(k+1)}$}} is uniquely determined as extension of the derivation {\small{$\delta^{-k}$}} to the whole {\small{$A^{-(k+1)}$}} by requiring that it acts as follows on the new variables:
$$\delta^{-(k+1)}(T_{i}^j)=  \tau_{i} \cdot j T_{i}^{j-1}.$$
Because, by construction, {\small{$A^{-(k+1)}$}} coincides with {\small{$A^{-k}$}} in degree higher than $-(k+1)$, it still holds that
$$H^{-j}(A^{-(k+1)}) = 
\left\lbrace \begin{array}{ll}
R/M & \quad  \mbox{if} \quad j=0\\
0 & \quad  \mbox{if} \quad j=1, \dots, k-1 \ 
\end{array}
\right.
$$
while in degree $-k$ we have that  
$$H^{-k}(A^{-(k+1)})= \frac{H^{-k}(A^{-k})}{ \langle \tau_1, \dots, \tau_{n_{k}}\rangle } = \{ 0 \}.$$
This is consequence of {\small{$\{ \tau_i \}$}} being a family of generators of {\small{$H^{-k}(A^{-k})$}} as an $R$-module. Hence {\small{$A^{-(k+1)}$}} gives a free and acyclic resolution of $R/M$ up to degree $-k$. Therefore, the algebra {\small{$A=\bigcup_{k=0}^{\infty} A^{-k}$}} is a free and acyclic resolution of $R/M$ such that {\small{$H^{0}(A) = R/M$}}.\\
\\
We observe that, by construction, the algebra {\small{$A$}} of a Tate resolution {\small{$(A, \delta)$}} can be viewed as a symmetric algebra
$$A=\Sym_{R}(\mathcal{W}^{*}_{T})$$ 
for some graded {\small{$R$}}-module {\small{$\mathcal{W}^{*}_{T}=\oplus_{j\leqslant -1}[\mathcal{W}_{T}^{*}]^{j}$}}, with locally-free and finitely-generated homogeneous components {\small{$[\mathcal{W}^{*}_{T}]^{j}$}}. This {\small{$R$}}-module {\small{$\mathcal{W}^{*}_{T}$}} collects all the formal variables $\{T_{i}\}$ introduced step by step in the Tate's algorithm.

\begin{oss}
The procedure described above to extend the algebra {\small{$A^{-k}$}} to the algebra {\small{$A^{-(k+1)}$}} in the case of an {\em even} degree $-(k+1)$ it is applicable only under the hypothesis that the ring $R$ contains a subfield of characteristic $0$. Since this hypothesis will be satisfied in the context where we want to apply Tate's algorithm, we restricted to this setting, which simplifies the construction (cf. \cite{Tate}). 
\end{oss}

\subsection{The extended configuration space} 
A key role in the construction of an extended theory {\small{$(\widetilde{X}, \widetilde{S})$}} associated to a given gauge theory {\small{$(X_{0}, S_{0})$}} is played by the so-called {\em generators of type $\beta$} of a given Tate resolution {\small{$(A, \delta)$}} of {\small{$J(S_{0})$}} on the ring {\small{$\mathcal{O}_{X_{0}}$}}.

 \begin{definition}
Given a Tate resolution {\small{$(A, \delta)$}}, the {\em generators of type $\beta$} are inductively defined as follows: all generators of {\small{$A^{-1}_{-1}$}} are of type $\beta$ while a generator {\small{$\gamma^{-q} \in A^{-q}_{-q}$}}, with $q>1$, is of type $\beta$ if there exists {\small{$\{ r_{j} \}$}}, with {\small{$j= 1, \dots, m_{j}$}}, a collection  of elements  of the ring $R$ such that
$$\delta(\gamma^{-q}) = r_{1}\beta^{-q+1}_{1} + r_{2} \beta^{-q+1}_{2} + \dots + r_{m_{j}} \beta^{-q+1}_{m_{j}}$$
with {\small{$\beta^{-q+1}_{1}, \beta^{-q+1}_{2}, \dots, \beta^{-q+1}_{m_{j}},$}} generators of type $\beta$ of degree $-q+1$. Thus for this generator {\small{$\gamma^{-q}$}} we use the notation {\small{$\beta^{-q}$}}. 
\end{definition}

Selecting the generators of type $\beta$ of a Tate resolution {\small{$(A, \delta)$}} of {\small{$J(S_{0})$}} on {\small{$\mathcal{O}_{X_{0}}$}} we determine the anti-fields/anti-ghost fields used to extend the space {\small{$X_{0}$}}. However, in order to comply with the requirement of {\small{$\widetilde{X}$}} being symmetric in the field/ghost field content on one hand and the anti-field/anti-ghost field content on the other hand, we have to impose that 
{\small{$(A, \delta)$}} is a Tate resolution of {\small{$J(S_{0})$}} on {\small{$\mathcal{O}_{X_{0}}$}} satisfying 
\begin{equation}
\label{condizione Tate resolution}
A^{-1}=T_{X_{0}}[1],
\end{equation}
with {\small{$T_{X_{0}}[1]$}} the shifted tangent space of {\small{$X_{0}$}}. Indeed, this property ensures that {\small{$\widetilde{X}$}} has in degree $-1$ the anti-fields corresponding to the fields in {\small{$X_{0}$}}. As a consequence, the algebra {\small{$A$}} is a symmetric algebra 
 $$
A=\Sym_{\mathcal{O}_{X_{0}}}(\mathcal{W}^{*}_{T}) \quad \quad \mbox{ with } \quad \quad \mathcal{W}_{T}^{*}= T_{X_{0}}[1]\oplus\mathcal{E}_{T}^{*}[1], 
$$
for {\small{$\mathcal{E}^{*}_{T}=\oplus_{i\leqslant - 1} [\mathcal{E}^{*}_{T}]^{i}$}} a graded {\small{$\mathcal{O}_{X_{0}}$}}-module with finitely-generated homogeneous components. Then, given a Tate resolution {\small{$(A, \delta)$}} satisfying condition \eqref{condizione Tate resolution}, let {\small{$\mathcal{W}^{*}= T_{X_{0}}[1]\oplus\mathcal{E}^{*}[1]$}} denote the {\small{$\mathcal{O}_{X_{0}}$}}-submodule of {\small{$\mathcal{W}_{T}^{*}$}} determined by generators of type $\beta$. Thus the corresponding extended configuration space {\small{$\widetilde{X}$}} is defined as follows:
$$\widetilde{X} = \mathcal{E}^{*}[1] \oplus T_{X_{0}}[1] \oplus X_{0} \oplus \mathcal{E}, $$
where the summand {\small{$T_{X_{0}}[1]$}} gives the anti-fields content of {\small{$\widetilde{X}$}}, {\small{$\mathcal{E}^{*}[1]$}} controls the anti-ghost fields content of {\small{$\widetilde{X}$}} and has been determined by the type-$\beta$ generators of a Tate resolution of {\small{$J(S_{0})$}} on {\small{$\mathcal{O}_{X_{0}}$}} and $\mathcal{E}$ is the shifted-dual module of {\small{$\mathcal{E}^{*}[1]$}} describing the ghost-fields content of {\small{$\widetilde{X}$}}.

 \begin{oss}
To determine a Tate resolution may require involved computations and force the introduction of an infinite number of new generators: indeed, while the hypothesis of $R$ being a Noetherian ring guarantees that at each step of the algorithm we introduce only a finite number of new formal variables, nothing ensures that the procedure has to stop after a finite number of steps. This would then give rise to an extended configuration space {\small{$\widetilde{X}$}} with independent ghost/anti-ghost fields in any degree. On the contrary, the generators of type $\beta$ can be easily computed by an inductive procedure and, if no redundant generators are introduced throughout the construction, they form a finite family, inducing an extended configuration space {\small{$\widetilde{X}$}} with a finite number of ghost/anti-ghost fields. Moreover, taking a full Tate resolution may cause the loss of meaning for concepts as the {\em level of reducibility $L$} of an extended theory: indeed, because $L$ is defined as $L=k-1$, where $k$ is the highest degree of a non-trivial homogeneous component in {\small{$\widetilde{X}$}}, if {\small{$\widetilde{X}$}} has non-trivial components in any degree, $L$ would not detect any properties of a gauge theory. On the other hand, considering only type-$\beta$ generators, $L$ still gives an estimate of the \textquotedblleft complexity\textquotedblright of the gauge symmetry considered (cf. Section \ref{Section: Extended varieties for a matrix model}). 
 \end{oss}

\subsection{The extended action}
We explain how to construct an extended action {\small{$\widetilde{S}$}} by a sequence of approximations, where at each step we add terms with increasing degree in the ghost fields. This procedure was inspired by the one proposed in \cite{felder}, the main difference laying in the fact of only consider the preselected family of type-$\beta$ generators. This allows to a more explicit description of the procedure. \\
\\
We begin by introducing the following notation: given a generic element {\small{$\varphi \in \mathcal{O}_{\widetilde{X}}$}}, we write it as as sum {\small{$ \varphi = \sum_{i \in I} \varphi_{n, i} \varphi_{p, i}, $}}
where 
$$ \varphi_{n, i} \in \Sym_{\mathcal{O}_{X_{0}}}(\mathcal{E}^{*}[1] \oplus T^{*}_{X_{0}}[1]\oplus X_{0}) \quad \quad \mbox{ and } \quad \quad  \varphi_{p, i} \in \Sym_{\mathcal{O}_{X_{0}}}(\mathcal{E}). $$
Then,  the {\em negative degree} and the {\em positive degree} of an element $\varphi$ are respectively defined as follows:
$$\deg_{n}(\varphi) = \max_{i \in I} \ \deg(\varphi_{n, i}) \in \mathbb{Z}_{\leqslant 0} \quad \mbox{ and } \quad  \deg_{p}(\varphi) = \min_{i \in I} \ \deg(\varphi_{p, i})\in \mathbb{Z}_{\geqslant 0},$$
while the {\em degree} of a summand {\small{$\varphi_{i}$}} is the sum of its negative and its positive degree:
$$\deg(\varphi_{i}) = \deg_{n}(\varphi_{i}) + \deg_{p}(\varphi_{i})  \in \mathbb{Z}.$$
Moreover, as useful tools in the following construction we consider the family of ideals {\small{$\{ F^{r}\mathcal{O}_{\widetilde{X}}\}$}}, for $r\geqslant 0 $, which determines a descending filtration of the graded algebra {\small{$\mathcal{O}_{\widetilde{X}}$}}:
$$\mathcal{O}_{\widetilde{X}} = F^{0}\mathcal{O}_{\widetilde{X}}\supseteq F^{1}\mathcal{O}_{\widetilde{X}} \supseteq F^{2}\mathcal{O}_{\widetilde{X}}\supseteq \dots$$
with
$$F^{r}\mathcal{O}_{\widetilde{X}} := \{ \varphi \in \mathcal{O}_{\widetilde{X}}: \ \deg_{p}( \varphi)\geqslant r \} \cup \{ 0\}.$$
While the ideals {\small{$F^{r}\mathcal{O}_{\widetilde{X}}$}} look at the positive degree of the elements in {\small{$\mathcal{O}_{\widetilde{X}}$}}, we introduce a family {\small{$\{ I_{\widetilde{X}}^{(q)}\}$}}, for  $q\geqslant 1$, of modules over {\small{$\mathcal{O}_{X_{0}}$}}, which consider the number of positively-graded generators in elements of {\small{$\mathcal{O}_{\widetilde{X}}$}}:  
$$I_{\widetilde{X}}^{(q)}: = \Big\{ \varphi= \sum_{I=(i_{1}, \dots, i_{q})} \varphi_{n, I} \beta_{i_{1}}\dots\beta_{i_{q}}: \varphi_{n, I} \in \Sym_{\mathcal{O}_{X_{0}}}(T_{X_{0}}[1] \oplus \mathcal{E}^{*}[1])\Big\}.$$
Notice that these modules are also closed with respect to the product for elements in {\small{$\mathcal{O}_{\widetilde{X}} / F^{1}\mathcal{O}_{\widetilde{X}}$}}. Finally, {\small{$\{I_{\widetilde{X}}^{\geqslant q}\}$}} is a collection of ideals, which are defined as unions of the {\small{$\mathcal{O}_{X_{0}}$}}-modules {\small{$I_{\widetilde{X}}^{(q)}$}}, i.e.
$$I_{\widetilde{X}}^{\geqslant q} = \bigcup_{s \geqslant q} I_{\widetilde{X}}^{(s)}, $$
for $q\geqslant 1$ and {\small{$S_{lin}$}} is the so-called {\em linear action}, which denotes the following sum:
$$S_{lin} = S_{0} + \sum_{k \in K} \delta(C^{*}_{k}) C_{k} + \sum_{j\in J} \delta(\beta^{*}_{j}) \beta_{j},$$
with {\small{$C^{*}_{k} \in [\mathcal{E}^{*}[1]]^{-2}$, $C_{k} \in [\mathcal{E}]^{1}$, $\beta^{*}_{j} \in [\mathcal{E}^{*}[1]]^{-(q+1)}$}} and {\small{$\beta_{j}\in [\mathcal{E}]^{q}$}}, for $q>1$. \\
We now present two technical lemmas whose statements have to be compared with \cite[Lemmas $4.3, 4.6$]{felder}. Although the statements have been partially modified to comply with the fact of considering only type-$\beta$ generators, because the proofs can be easily adapted to this different setting, we refer to \cite{felder} for further details. 

\begin{lemma}
\label{lemma tecnico 1}
The $1$-degree differential operator 
$$\Phi:=\left\lbrace S_{lin}, - \right\rbrace: \mathcal{O}^{\bullet}_{\widetilde{X}} \longrightarrow \mathcal{O}^{\bullet +1}_{\widetilde{X}}$$
coincides with the operator $\delta \otimes Id$ modulo {\small{$F^{1}\mathcal{O}_{\widetilde{X}}$}}, for $\delta$ the differential operator of the fixed Tate resolution {\small{$(A, \delta)$}}, restricted to act only on generators of type $\beta$. Therefore, given an element {\small{$\varphi\in \mathcal{O}_{\widetilde{X}}$}}, with {\small{$\varphi = \sum_{i} \varphi_{n, i} \varphi_{p, i}$}}, it holds:
$$\Phi(\varphi):= \left\lbrace S_{lin}, \varphi \right\rbrace = \sum_{i} \delta(\varphi_{n, i}) \cdot Id(\varphi_{p, i}) \quad \quad (mod \mbox{ } F^{1}\mathcal{O}_{\widetilde{X}}). $$
\end{lemma}

\begin{lemma}
\label{lemma tecnico 2}
Given an element {\small{$\varphi \in F^{q}\mathcal{O}_{\widetilde{X}}$}}, $q\geqslant 0$, then  
$$\{ \varphi, \psi\} \in F^{q+1}\mathcal{O}_{\widetilde{X}}, \quad \quad  \mbox{ for } \quad \quad  \psi \in I^{\geqslant 2}_{\widetilde{X}} \cap \mathcal{O}^{0}_{\widetilde{X}}. $$
Moreover, if in particular {\small{$\varphi \in F^{q}\mathcal{O}_{\widetilde{X}}^{0}$}}, then
$$\{ \varphi, \psi\} \in F^{q}\mathcal{O}_{\widetilde{X}}^{1},   \mbox{ for }  \psi \in \mathcal{O}_{\widetilde{X}}^{0} \quad \quad \mbox{ and } \quad \quad  \{ \varphi, \psi\} \in F^{q+1}\mathcal{O}_{\widetilde{X}}^{1} ,  \mbox{ for }  \psi \in F^{q}\mathcal{O}_{\widetilde{X}}^{0}.$$ 
\end{lemma}

The last element that we have to introduce before being able to present the construction of the extended action $\widetilde{S}$ is the cohomology complex $(\mathcal{G}_{q, r}^{\bullet}, d)$, with $r\leqslant  q$ fixed and non-negative.

\begin{definition}
\label{complesso per algoritmo}
Let $q$, $r$ be two fixed values in {\small{$\mathbb{Z}_{\geqslant0}$}}, with $r\leqslant q$. The pair {\small{$(\mathcal{G}_{q, r}^{\bullet}, d)$}} is a collection of sets and a graded map on them, where
 $$\mathcal{G}_{q, r}^{j}= \pi_{q}(F^{q}\mathcal{O}_{\widetilde{X}}^{j} \cap I_{\widetilde{X}}^{(r)})\cup \{ 0 \},$$
 for $j\leqslant q$, with $\pi_{q}$ the canonical projection {\small{$\pi_{q}: F^{q}\mathcal{O}_{\widetilde{X}} \rightarrow F^{q}\mathcal{O}_{\widetilde{X}}/ F^{q+1}\mathcal{O}_{\widetilde{X}}$}}. Thus:
 $$\mathcal{G}_{q, r}^{j} = \Big\{ \varphi =\sum_{i \in I} \varphi_{n, i}\varphi_{p, i} \in \mathcal{O}^{j}_{\widetilde{X}} : \ \deg(\varphi_{p, i}) = q, \varphi_{p, i} = \beta_{j_{1}} \cdots \beta_{j_{r}}, \ \forall i \Big\} \cup \{ 0 \}.$$
Concerning the graded map {\small{$d= \{ d^{j} \}_{j\leqslant p}$}} , with {\small{$d^{j}: \mathcal{G}_{q, r}^{j} \rightarrow \mathcal{G}_{q, r}^{j+1}\ ,$}} it is defined as  
 $$d(\varphi) = (\delta \otimes Id)(\varphi) = \sum_{i\in I} \delta(\varphi_{n, i})\varphi_{p, i},$$
 for {\small{$\varphi\in \mathcal{G}_{q, r}^{j}$}}, where $\delta$ is the coboundary operator of the fixed Tate resolution.
\end{definition}

\begin{prop}
\label{prop per complesso per algoritmo}
The pair {\small{$(\mathcal{G}_{q, r}^{\bullet}, d)$}} defines a cochain complex. 
\end{prop}

\begin{proof}
The statement follows by noticing that, because the operator $\delta$ acts only on the generators of non-positive degree, the map $d$, applied to an element $\varphi$, not only preserves its positive degree but also the number of positively-graded generators appearing in $\varphi$. This allows to conclude that $d$ is well defined as map from {\small{$\mathcal{G}_{q, r}^{j}$}} to {\small{$\mathcal{G}_{q, r}^{j+1}$}}. All the remaining requirements on $d$, among which the condition of defining a differential operator, that is, of satisfying {\small{$d^{j+1}\circ d^{j}\equiv 0$}}, can be immediately deduced from the analogous properties of the coboundary operator $\delta$. 
\end{proof}
\noindent
We now have everything needed to state the main theorem on the existence of an extended action {\small{$\widetilde{S} \in \mathcal{O}_{\widetilde{X}}$}}, which solves the classical master equation on {\small{$\mathcal{O}_{\widetilde{X}}$}}, for {\small{$\widetilde{X}$}} the already determined extended configuration space.

\begin{theorem}
\label{th_action, existence}
Given a gauge theory {\small{$(X_{0}, S_{0})$}}, with {\small{$X_{0}$}} a real affine variety and {\small{$S_{0}\in \mathcal{O}_{X_{0}}$}}, and a Tate resolution {\small{$(A, \delta)$}} of {\small{$J(S_{0})$}} on {\small{$\mathcal{O}_{X_{0}}$}} satisfying condition \eqref{condizione Tate resolution}, let {\small{$\widetilde{X}$}} be the extended configuration space determined by the type-$\beta$ generators of {\small{$A$}}. If the induced cohomology complex {\small{$(\mathcal{G}_{q, r}, d)$}} is such that
\begin{equation}
\label{eq: condizione sufficiente per esistenza ext action}
H^{j}(\mathcal{G}_{q, r}, d) = 0, \quad \quad \mbox{ for } j\leqslant q, 
\end{equation}
there exists a function {\small{$\widetilde{S} \in \mathcal{O}^{0}_{\widetilde{X}}$}} which solves the classical master equation on {\small{$\mathcal{O}_{\widetilde{X}}$}}, that is, such that {\small{$\{\widetilde{S}, \widetilde{S}\}=0$}}, and which satisfies {\small{$\widetilde{S}|_{X_{0}}=S_{0}$}}, with
$$\widetilde{S} =S_{lin}:= S_{0} + \sum_{k \in K} \delta(C^{*}_{k}) C_{k} + \sum_{j\in J} \delta(\beta^{*}_{j}) \beta_{j} \quad (\mbox{mod } I_{\widetilde{X}}^{\geqslant 2}).$$
\end{theorem}
\noindent
Once again, the statement of this theorem has to be compared with \cite[Theorem $4.5$]{felder}, the main difference lying in having selected only type-$\beta$ generators in the Tate resolution. Hence condition (\ref{eq: condizione sufficiente per esistenza ext action}) has to be explicitly imposed while, if we consider a complete (eventually infinite) Tate resolution, the vanishing of the cohomology groups {\small{$H^{j}(\mathcal{G}_{q, r}, d)$}} for {\small{$j\leqslant q$}} is a direct consequence of the Tate resolution defining an acyclic complex. However, this extra requirement is only a {\em sufficient} condition to ensure the existence of {\small{$\widetilde{S}$}} (cf. Remark \ref{oss: condizione per poter applicare l'algoritmo per l'azione}). 

\begin{proof}
To prove the existence of an action {\small{$\widetilde{S}$}} as claimed in the statement, we present an algorithm to construct it. This method determines {\small{$\widetilde{S}$}} via a sequence of approximations obtained by introducing terms of increasing positive degree. As first step of this algorithm, we define {\small{$\widetilde{S}_{\leqslant 1}$}} as follows:
$$\widetilde{S}_{\leqslant 1}= S_{0} + \sum_{k \in K} \delta(C^{*}_{k}) C_{k}  + \sum_{j\in J} \delta(\beta^{*}_{j}) \beta_{j}\ ,$$
where {\small{$\{ C_{i} \}$}} generate {\small{$[\mathcal{E}]^{1}$, $\{ \beta_{j} \}$}} collectively denote the generators of {\small{$[\mathcal{E}]^{q}$}}, for {\small{$q>1$}}, and {\small{$\{ C^{*}_{i}\}$}},{\small{ $\{ \beta^{*}_{j} \}$}} are the dual generators of {\small{$\{ C_{i} \}$}} and {\small{$\{ \beta_{j} \}$}} respectively. Then we compute the quantity {\small{$\big\{ \widetilde{S}_{\leqslant 1}, \widetilde{S}_{\leqslant 1}\big\}$}}. If this quantity is zero, {\small{$\widetilde{S}_{\leqslant 1}$ }}is a solution of the classical master equation on {\small{$\mathcal{O}_{\widetilde{X}}$}} and the algorithm stops. Indeed, by definition {\small{$\widetilde{S}_{\leqslant 1}$}} automatically satisfies the other two required properties, that is, 
$$\widetilde{S}_{\leqslant 1}|_{X_{0}} = S_{0},  \quad \quad \quad \mbox{ and } \quad \quad \quad \widetilde{S}_{\leqslant 1} = S_{lin}\quad (\mbox{mod } I_{\widetilde{X}}^{\geqslant 2}).$$
Hence we define {\small{$ \widetilde{S}:=\widetilde{S}_{\leqslant 1}$}} and conclude the construction. Otherwise, if {\small{$\big\{ \widetilde{S}_{\leqslant 1}, \widetilde{S}_{\leqslant 1}\big\} \neq 0$}}, we have to apply the generic step of the algorithm for {\small{$q=1$}}, to obtain the approximation of the extended action up to positive degree $2$. However, in order to implement this step of the algorithm, we first have to notice that the approximated action {\small{$\widetilde{S}_{\leqslant 1}$}} satisfies the following properties:
\begin{enumerate}
\item $\widetilde{S}_{\leqslant 1}|_{X_{0}} = S_{0}$;
\vspace{1.2mm}
\item $\widetilde{S}_{\leqslant 1} = S_{lin}, \quad $ (mod $I_{\widetilde{X}}^{\geqslant 2}$);
\vspace{1.2mm}
\item $\{ \widetilde{S}_{\leqslant 1}, \widetilde{S}_{\leqslant 1}\} \in  I_{\widetilde{X}}^{\geqslant 2} \cap F^{2}\mathcal{O}^{0}_{\widetilde{X}}$;
\end{enumerate}
where, being in degree $2$, {\small{$I_{\widetilde{X}}^{\geqslant 2} \cap F^{2}\mathcal{O}^{0}_{\widetilde{X}}$}} happens to coincide with {\small{$I_{\widetilde{X}}^{\geqslant 2}$}}. While the first two properties follows immediately from the explicit expression of the function {\small{$\widetilde{S}_{\leqslant 1}$}}, the last condition can be verified by an explicit computation:
\begin{multline}
\label{computation approx action1}
\{ \widetilde{S}_{\leqslant 1}, \widetilde{S}_{\leqslant 1}\} =  \  2 \sum_{k} \big\{ S_{0}, \delta(C^{*}_{k}) C_{k} \big\}  + \sum_{k, l} \big\{ \delta(C^{*}_{k}) C_{k}, \delta(C^{*}_{l}) C_{l} \big\} \\
+ 2\sum_{k, j} \big\{ \delta(C^{*}_{k}) C_{k}, \delta(\beta^{*}_{j}) \beta_{j} \big\} +\sum_{j, m} \big\{ \delta(\beta^{*}_{j}) \beta_{j}, \delta(\beta^{*}_{m}) \beta_{m} \big\},  
\end{multline}
where we use the fact that, due to the definition of the Poisson bracket, the only non-trivial contribution to the above expression involving the initial action {\small{$S_{0}$}} is the one containing the terms {\small{$\delta(C^{*}_{k})$}}, that is, the only summands that depend on the anti-fields {\small{$x^{*}_{i}$}}. We consider separately the different components of the sum in (\ref{computation approx action1}). Concerning the first term, we have that, being {\small{$\delta(C^{*}_{k}) = \sum_{m}  g_{km}(x)x^{*}_{m}$}}, for certain {\small{$g_{km} \in \mathcal{O}_{X_{0}}$}}, then
$$\sum_{k} \big\{ S_{0}, \delta(C^{*}_{k}) C_{k} \big\} = \sum_{k, m} \delta(x^{*}_{m})g_{km}(x)C_{k} = \sum_{k} \delta(\delta(C^{*}_{k}))C_{k}=0,$$
where we use the fact that $\delta$ is a coboundary operator and that, as a consequence of the algorithmic construction of $\delta$ as differential in the Tate resolution of the Jacobian ring {\small{$J(S_{0})$}}, it holds the following identity
$$\{S_{0}, x^{*}_{m} \}= \partial_{m}S_{0} = \delta(x^{*}_{m}).$$
Once again recalling the explicit expression of {\small{$\delta(C^{*}_{k})$}}, we deduce that the second summand in (\ref{computation approx action1}) satisfies
$$\sum_{k, l} \big\{ \delta(C^{*}_{k}) C_{k}, \delta(C^{*}_{l}) C_{l} \big\} = - \sum_{k, l} C_{k} \big\{ \delta(C^{*}_{k}), \delta(C^{*}_{l}) \big\}  C_{l} = 0 \quad \quad (\mbox{mod } I_{\widetilde{X}}^{\geqslant 2}).$$
To deduce that also the last two summands in (\ref{computation approx action1}) contribute to the sum by elements in $I_{\widetilde{X}}^{\geqslant 2}$, we use the characterization of type-$\beta$ generators, which allows to deduce that the only terms of type {\small{$\delta(\beta^{*}_{j})$}} depending on the anti-ghost fields {\small{$C^{*}_{i}$}} are the ones referring to generators {\small{$\beta^{*, -2}_{j}$}} of degree $-2$. In particular, these generators satisfy 
$$\delta(\beta^{*, -2}_{j}) = \sum_{m} f_{mj}C^{*}_{m} \quad \quad \mbox{ with }  \quad \quad \sum_{m} f_{mj}\delta(C^{*}_{m})=0$$
for {\small{$f_{mj} \in \mathcal{O}_{X_{0}}$}}. Therefore, we deduce the following series of equalities:
\begin{equation}
\label{eq con uso di tipo beta generators}
\begin{array}{l}
\sum_{k, j} \big\{ \delta(C^{*}_{k}) C_{k}, \delta(\beta^{*, -2}_{j}) \beta^{1}_{j} \big\}\\
[1.5ex]
\quad \quad =  \sum_{k, j,m} \delta(C^{*}_{k}) \big\{  C_{k}, f_{mj} C^{*}_{m} \big\} \beta^{1}_{j} -
\sum_{k, j} C_{k}  \big\{ \delta(C^{*}_{k}), \delta(\beta^{*, -2}_{j})\big\} \beta^{1}_{j} \\
[1.5ex] 
\quad \quad = \sum_{k,j}  f_{kj}\delta(C^{*}_{k}) \beta^{1}_{j}=  0 \hspace{30mm} (\mbox{mod } I_{\widetilde{X}}^{\geqslant 2}). 
\end{array}
\end{equation}
The specific characteristic that distinguishes the type-$\beta$ generators enters these series of equalities at the last step, allowing to deduce that the last sum is zero. A completely analogous series of equalities can be written about the last summand in (\ref{computation approx action1}), allowing to draw the conclusion that 
$$\{ \widetilde{S}_{\leqslant 1}, \widetilde{S}_{\leqslant 1}\} = 0 \hspace{20mm} (\mbox{mod } I_{\widetilde{X}}^{\geqslant 2}),$$
and therefore, all the three properties claimed for the approximated action {\small{$\widetilde{S}_{\leqslant 1}$}} are verified. Thus the generic step of the algorithm can be applied to the approximated action {\small{$\widetilde{S}_{\leqslant 1}$}} to determine {\small{$\widetilde{S}_{\leqslant 2}$}}, i.e., the approximated action up to positive degree {\small{$2$}}.\\
\\
\noindent
{\em The generic step.} To construct an approximation of the action {\small{$\widetilde{S}$}} up to degree {\small{$q+1$}} in the ghost fields, consider the approximation {\small{$\widetilde{S}_{\leqslant q}$}}, obtained in the previous step of the algorithm. Because of the inductive hypothesis, the approximated action {\small{$\widetilde{S}_{\leqslant q}$}} fulfils the following conditions:
\begin{enumerate}
\item $\widetilde{S}_{\leqslant q}|_{X_{0}} = S_{0}$;
\vspace{1mm}
\item $\widetilde{S}_{\leqslant q} = S_{lin}, \quad $ (mod $I_{\widetilde{X}}^{\geqslant 2}$);
\vspace{1mm}
\item $\{ \widetilde{S}_{\leqslant q}, \widetilde{S}_{\leqslant q}\} \in  I_{\widetilde{X}}^{\geqslant 2} \cap F^{q+1}\mathcal{O}^{0}_{\widetilde{X}}$.
\end{enumerate}
Moreover we can verify that 
$$\big\{ \widetilde{S}_{\leqslant q}, \widetilde{S}_{\leqslant q}\big\} \in \bigoplus_{r=2}^{q+1} \Ker(\mathcal{G}^{1}_{q+1, r}, d) \quad \quad (\mbox{mod } \ F^{q+2}\mathcal{O}_{\widetilde{X}}).$$
Indeed, this conclusion can be easily drawn from the following series of equalities:
$$0= \big\{ \widetilde{S}_{\leqslant q},\big\{ \widetilde{S}_{\leqslant q}, \widetilde{S}_{\leqslant q}\big\}\big\} =  \big\{ S_{lin},\big\{ \widetilde{S}_{\leqslant q}, \widetilde{S}_{\leqslant q}\big\}\big\} = (\delta \otimes Id)\big\{ \widetilde{S}_{\leqslant q}, \widetilde{S}_{\leqslant q}\big\}.$$
While the first equality is consequence of the graded Jacobi identity satisfied by the Poisson bracket, the second passage results from applying Lemma \ref{lemma tecnico 2}: indeed, by inductive hypothesis the quantity {\small{$\{ \widetilde{S}_{\leqslant q}, \widetilde{S}_{\leqslant q}\big\}$}} is an element in {\small{$F^{q+1}\mathcal{O}_{\widetilde{X}}$}} as well as {\small{$\widetilde{S}_{\leqslant q} - S_{lin}$}} belongs to {\small{$I_{\widetilde{X}}^{\geqslant 2} \cap \mathcal{O}_{\widetilde{X}}^{0}$}} and hence the contribution coming from the Poisson bracket computed on {\small{$\widetilde{S}_{\leqslant q} - S_{lin}$}} and {\small{$\{ \widetilde{S}_{\leqslant q}, \widetilde{S}_{\leqslant q}\big\}$}} has positive degree at least {\small{$q+2$}}, that is, is zero modulo {\small{$F^{q+2}\mathcal{O}_{\widetilde{X}}$}}. Finally, the last passage of the above equation follows from Lemma \ref{lemma tecnico 1}. \\
The hypothesis of trivial cohomology for the cohomology complex {\small{$(\mathcal{G}_{q, r}, d)$}} ensures the existence of an element {\small{$\nu \in I^{\geqslant 2}_{\widetilde{X}}\cap F^{q+1}\mathcal{O}^{0}_{\widetilde{X}}$}} such that
\begin{equation}
\label{eq: def di nu}
(\delta \otimes Id)(\nu) = -\frac{1}{2}\big\{ \widetilde{S}_{\leqslant q}, \widetilde{S}_{\leqslant q}\big\} \quad \quad \quad (\mbox{mod } \ F^{q+2}\mathcal{O}_{\widetilde{X}}) .
\end{equation}
Hence, the approximation {\small{$\widetilde{S}_{\leqslant q+1}$}} up to positive degree {\small{$q+1$}} is defined to be
$$\widetilde{S}_{\leqslant q+1}:= \widetilde{S}_{\leqslant q}+ \nu, $$
which satisfies the three properties required for the inductive procedure: indeed, while the first two properties immediately follows from the inductive hypothesis on {\small{$\widetilde{S}_{\leqslant q}$}} and from the element $\nu$ belonging to {\small{$I^{\geqslant 2}_{\widetilde{X}}\cap F^{q+1}\mathcal{O}^{0}_{\widetilde{X}}$}}, the third one can be checked as follows:
$$\big\{ \widetilde{S}_{\leqslant q+1}, \widetilde{S}_{\leqslant q+1}\big\} =  \big\{ \widetilde{S}_{\leqslant q}, \widetilde{S}_{\leqslant q}\big\} + 2 \big\{ S_{lin}, \nu \big\} = 0 \quad \quad (\mbox{mod } \ F^{q+2}\mathcal{O}_{\widetilde{X}})$$
where we first use Lemma \ref{lemma tecnico 2}, which ensures that the contribution coming from the Poisson bracket between {\small{$\widetilde{S}_{\leqslant q} - S_{lin}$}} and $\nu$ belongs to {\small{$F^{q+2}\mathcal{O}_{\widetilde{X}}$}} as well as the quantity {\small{$\{ \nu, \nu \}$}}, while the last equality is a direct consequence on the way how the element $\nu$ has been chosen. Therefore, the approximated action {\small{$\widetilde{S}_{\leqslant q+1}$}} satisfies all the properties required for the inductive construction, allowing to conclude the proof.
\end{proof}

\begin{oss}
 \label{oss: condizione per poter applicare l'algoritmo per l'azione}
 The requirement in (\ref{eq: condizione sufficiente per esistenza ext action}) concerning the vanishing of the cohomology groups {\small{$H^{j}(\mathcal{G}_{q, r}, d)$}} for {\small{$j\leqslant q$}} has been used in (\ref{eq: def di nu}) to ensure the possibility of defining the element {\small{$\nu$}} as an element that describes the cocycle {\small{$\{ \widetilde{S}_{\leqslant q}, \widetilde{S}_{\leqslant q}\}$}} as a coboundary element. However, the vanishing of all the cohomology groups (cf. in \cite{felder}) is a {\em sufficient} condition: to define {\small{$\nu$}} it is not necessary to have that every cocycle of degree {\small{$j$}} in the cohomology complex {\small{$(\mathcal{G}_{q, r}, d)$}}, with {\small{$j\leqslant q$}}, is also a coboundary element but it is enough to have that this property holds for the cocycle {\small{$\{ \widetilde{S}_{\leqslant q}, \widetilde{S}_{\leqslant q}\}$}} of interest. Next to it, the peculiar form of the selected type-$\beta$ generators plays a role in (\ref{eq con uso di tipo beta generators}): indeed, the condition of being type-$\beta$ generators is exactly the property that ensures the quantity in Equation (\ref{eq con uso di tipo beta generators}) being zero, it is the minimal request that accords the possibility of applying the algorithmic construction of an approximated action up to positive degree $2$ starting with the linear approximation {\small{$\widetilde{S}_{\leqslant 1}$}}. Hence it is the algorithmic procedure that tells us how to select the minimal number of generators to be able to proceed with the construction of an action $\widetilde{S}$, that is, the type-$\beta$ generators.
\end{oss}

Also in this context it is possible to introduce a notion of {\em gauge equivalence} of extended actions.

\begin{definition}
\label{gauge equivalent felder}
Let {\small{$\widetilde{X}$}} be an extended configuration space for a gauge theory {\small{$(X_{0}, S_{0})$}}. Given the Lie algebra {\small{$\mathfrak{g}_{\widetilde{X}}= \mathcal{O}_{\widetilde{X}}^{-1} \cap I_{\widetilde{X}}^{\geqslant 2}$}}, the {\em group of gauge equivalences} {\small{$G(\widetilde{X})$}} is defined as the following group of Poisson automorphisms: 
 $$G(\widetilde{X}) = exp(ad(\mathfrak{g}_{\widetilde{X}})).$$
\end{definition}
 
 \begin{theorem}
 \label{th_action, uniqueness}
 Let {\small{$(\widetilde{X}, \widetilde{S})$}} be an extended theory associated to a Tate resolution {\small{$(A, \delta)$}}, corresponding to a gauge theory {\small{$(X_{0}, S_{0})$}}. Suppose that the cohomology complex {\small{$(\mathcal{G}_{q, r}, d)$}} has vanishing cohomology groups {\small{$H^{j}(\mathcal{G}_{q, r}, d)$}} for {\small{$j\leqslant q$}}. If {\small{$\widetilde{S}^{\prime}\in \mathcal{O}_{\widetilde{X}}$}} is another solution of the classical master equation on {\small{$\mathcal{O}_{\widetilde{X}}$}}, with
$$\widetilde{S}^{\prime} = \widetilde{S} = S_{0} + \sum_{i} \delta(C^{*}_{i})C_{i} + \sum_{j \in J} \delta(\beta^{*}_{j}) \beta_{j} \quad (mod \ I^{\geqslant 2}_{\widetilde{X}})$$
then there exists a gauge equivalence {\small{$g \in G(\widetilde{X})$}} such that {\small{$\widetilde{S}^{\prime} = g \cdot \widetilde{S}.$}}
 \end{theorem}
 
Because the proof of \cite[Theorem $4.5$]{felder} can be easily adapted to this context, we refer to that pages for further details.  
 
 \begin{oss}
Once again, the requirement that the cohomology groups {\small{$H^{j}(\mathcal{G}_{q, r}, d)$}} vanish for {\small{$j\leqslant q$}} is a sufficient but not necessary condition to draw the conclusion claimed in Theorem \ref{th_action, uniqueness}. Indeed, to complete the proof of the statement it is enough that the cocycle {\small{$\widetilde{S}^{\prime} - \widetilde{S}$}} (mod {\small{$F^{p+1}\mathcal{O}_{\widetilde{X}}$}}) is in particular also a coboundary element in the cohomology complex {\small{$(\mathcal{G}_{q, r}, d)$}}, for any {\small{$p\geqslant 2$}}. The condition of the vanishing of the cohomology groups {\small{$H^{j}(\mathcal{G}_{q, r}, d)$}} is always satisfied when we consider all the generators introduced in a Tate resolution, which by definition is required to be acyclic. Having selected only the type-$\beta$ generators, this vanishing condition had to be explicitly inserted to ensure the validity of the statement.
 \end{oss}

\section{Application to a $U(2)-$matrix model}
\label{Section: Extended varieties for a matrix model}
\noindent
For a matrix model {\small{$(X_{0}, S_{0})$}} with a {\small{$U(2)$}}-gauge symmetry, we describe in detail how to implement the BV construction explained in Section \ref{Section: Construction of extended varieties}, arriving to determine the minimal extended theory {\small{$(\widetilde{X}, \widetilde{S})$}} corresponding to {\small{$(X_{0}, S_{0})$}}. Despite of its low dimension, this model turns out to be surprisingly rich and it gives interesting insights for the analysis of matrix models of higher order. Note that an application of the BV formalism to a $U(2)$-matrix model can also be found in \cite{beringgrosse}, where a different approach is followed and different goals are pursued.
 
 \subsection{A $\mathbf{U(2)}$-matrix model}
 Let {\small{$(X_{0}, S_{0})$}} be a gauge theory where the configuration space {\small{$X_{0}$}} is the real affine variety of $2 \times 2$ self-adjoint matrices, i.e.,
 $$X_{0}= \left\lbrace M \in M_{2}(\mathbb{C}): M^{*} = M \right\rbrace, $$
 and where the action functional {\small{$S_{0}: X_{0}\rightarrow \mathbb{R}$}} is supposed to be a regular function on $X_{0}$, that is, {\small{$S_{0} \in \mathcal{O}_{X_{0}}$}}, invariant under the action {\small{$F: \mathcal{G} \times X_{0} \rightarrow X_{0}$}} of the gauge group {\small{$\mathcal{G} = U(2)$}} on {\small{$X_{0}$}} by conjugation:
$$F(U, M) = UMU^{*},$$
for {\small{$U \in U(2)$}} and {\small{$M \in X_{0}$}}. The implementation of the BV construction on the pair {\small{$(X_{0}, S_{0})$}} requires an explicit expression of the action {\small{$S_{0}$}} in terms of the coordinates on {\small{$X_{0}$}}. Thus we fix a basis for {\small{$X_{0}$}} given by the Pauli matrices (together with the identity matrix):
 $$
\sigma_{1}= \begin{pmatrix}
		    0 & 1\\
		    1 & 0
	\end{pmatrix}
, \quad 
\sigma_{2}= \begin{pmatrix}
		    0 & -i\\
		    i & 0
		\end{pmatrix}
, \quad 
\sigma_{3}= \begin{pmatrix}
		    1 & 0\\
		    0 & -1
		\end{pmatrix}	 
, \quad
\sigma_{4}= \begin{pmatrix}
		    1 & 0\\
		    0 & 1
		\end{pmatrix}.
$$
Hence {\small{$X_{0}$}} is isomorphic to a $4$-dimensional real vector space generated by four independent initial fields:
$$X_{0} \cong \langle M_{1}, M_{2}, M_{3}, M_{4} \rangle_{\mathbb{R}},$$
with {\small{$\{ M_{a}\}$}}, for $a=1, \dots, 4$, the dual basis of {\small{$\{ \sigma_{a}\}$}}, $a=1, \dots, 4$. In this set of coordinates, the ring of regular functions on {\small{$X_{0}$}} is the ring of polynomials in the real variables {\small{$M_{a}$}}, {\small{$\mathcal{O}_{X_{0}} = \Pol_{\mathbb{R}}(M_{a})$}}, and the most generic form for a functional {\small{$S_{0}$}} on {\small{$X_{0}$}} that is invariant under the adjoint action of the gauge group {\small{$U(2)$}} is as symmetric polynomial in the eigenvalues {\small{$\lambda_{1}, \lambda_{2}$}} of the variable {\small{$M \in X_{0}$}} or, equivalently, as polynomial in the symmetric elementary polynomials $a_{1} = \lambda_{1} + \lambda_{2}$ and $a_{2} = \lambda_{1}\lambda_{2}$. Therefore, the action $S_{0}$ has the following form in terms of the coordinates $M_{a}$:
\begin{equation}
S_{0} = \sum_{k=0}^{r} \mbox{ }(M_{1}^2 + M_{2}^2 + M_{3}^2)^k \mbox{ } g_{k}(M_4)
\label{S_0 generale}
\end{equation}
with $g_{k}(M_{4})\in \Pol_{\mathbb{R}}(M_4)$ and for
$$ \lambda_i= M_{4}\pm \sqrt{M_{1}^2+M_{2}^2+M_{3}^2},  \quad a_{1}= 2M_4, \quad a_{2}= M^{2}_{4}- (M_{1}^2+M_{2}^2+M_{3}^2)\ .$$
We immediately notice that for any initial action $S_{0}$ of the form in (\ref{S_0 generale}), the partial derivatives  with respect to the variables $M_a$ satisfy the following linear relations over the ring $\mathcal{O}_{X_{0}}$:  
{\small{ $$ M_{1}({\partial_{M_{2}}}S_{0}) =M_{2}({\partial_{M_{1}}}S_{0}), \quad 
M_{1}({\partial_{M_{3}}}S_{0})=M_{3}({\partial_{M_{1}}}S_{0}), \quad 
 M_{2}({\partial_{M_{3}}}S_{0})=M_{3}({\partial_{M_{2}}}S_{0}). 
$$}}

\subsection{The minimal extended theory}
\label{BV-var model}
To determine the minimal extended theory {\small{$(\widetilde{X}, \widetilde{S})$}} for the gauge theory {\small{$(X_{0}, S_{0})$}} described above, we first concentrate on the construction of {\small{$\widetilde{X}$}}, arriving to the following theorem.

\begin{theorem}
\label{Theorem: extended theory for the U(2) model}
Given a  gauge theory {\small{$(X_{0}, S_{0})$}} with configuration space {\small{$X_{0}\simeq \langle M_a \rangle_{\mathbb R}$}} for $a=1, \dots, 4$, and action functional {\small{$S_{0}\in \mathcal{O}_{X_{0}}$}} of the form \eqref{S_0 generale}, the corresponding minimally-extended configuration space {\small{$\widetilde{X}$}} is a  $\mathbb{Z}$-supergraded real vector space whose explicit form depends on the action {\small{$S_{0}$}} as follows:
\begin{enumerate}
\item If {\small{$S_{0} \in \Pol_{\mathbb{R}}(M_{4})$}}, 
{\small{$$\widetilde{X} = X_{0} \oplus \langle M^{*}_{1}, \dots , M^{*}_{4}\rangle_{-1}.$$}}
\item If {\small{$GCD(\partial_{1}S_{0},  \partial_{2}S_{0}, \partial_{3}S_{0}, \partial_{4}S_{0}) = 1$}}, 
{\small{$$\hspace{13mm}\widetilde{X} = \langle E^{*} \rangle_{-3} \oplus \langle C^{*}_{1}, \cdots, C^{*}_{3}\rangle_{-2} \oplus \langle M^*_{1}, \dots, M^*_{4}\rangle_{-1} \oplus X_{0} \oplus \langle C_{1}, \cdots, C_{3} \rangle_{1} \oplus \langle E\rangle_{2} .
$$}}
\item If {\small{$GCD(\partial_{1}S_{0}, \partial_{2}S_{0}, \partial_{3}S_{0}, \partial_{4}S_{0}) = D \notin \mathbb{R}$}}, 
\vspace{1mm}\\
{\small{$\begin{array}{ll}
\widetilde{X} = & \langle K^* \rangle_{-4} \oplus \langle E^{*}_{1}, \dots, E^*_4\rangle_{-3} \oplus \langle C^{*}_{1}, \cdots, C^{*}_{6}\rangle_{-2} \oplus \langle M^*_{1}, \dots, M^*_{4}\rangle_{-1}\\
[.8ex]
&  \oplus X_{0} \oplus \langle C_{1}, \cdots, C_{6} \rangle_{1} \oplus \langle E_{1}, \dots, E_4\rangle_{2} \oplus \langle K\rangle_{3}.
\end{array}$}}
\end{enumerate} 
\end{theorem}

\begin{proof}
We recall that the strategy to construct {\small{$\widetilde{X}$}} is by extending the initial configuration space {\small{$X_{0}$}} via the introduction of anti-fields/anti-ghost fields of decreasing degree, starting in degree $-1$. Because the anti-fields in degree $-1$ are determined by the fields in {\small{$X_{0}$}}, we have the following exact sequence:
$$
A^{-1}_{-1}:=\Pol(M_a)\langle M_a^*\rangle
 \xrightarrow{\delta^{-1}_{-1}}
\Pol(M_a)
 \xrightarrow{\pi}
 J(S_0)  \rightarrow 
0\ ,
$$
where $\pi$ is the projection on the quotient and {\small{$M^{*}_{a}$}}, for $a=1, \dots, 4$, are Grassmannian variables of ghost degree $-1$, used to extended the ring {\small{$\mathcal{O}_{X_{0}} = \Pol_{\mathbb{R}}(M_{a})$}} to a $\mathbb{Z}_{\leqslant 0}$-graded algebra. Concerning the coboundary operator $\delta$, it is completely determined by imposing that it acts as follows on the generators {\small{$M_{a}^{*}$}}:
 $$\delta^{-1}_{-1}(M_a^*)=\partial_a S_{0}$$ 
 and then extending it by linearity on {\small{$\mathcal{O}_{X_{0}}$}}. To establish what is the minimal number of variables that need to be introduced in degree $-2$, we have to find a set of generators for the cohomology group {\small{$H^{-1}(A^{-1})=\Ker(\delta^{-1}_{-1})/\Imag(\delta^{-1}_{-2})\ .$}}
By a direct computation one can check that the result depends on the partial derivatives of the action {\small{$S_{0}$}} being or not coprime. Thus we analyze the different cases separately.\\
\\
\noindent
{\em Case 1}: Suppose that {\small{$S_{0}\in \Pol_{\mathbb{R}}(M_{4})$}}. Then, because {\small{$\Ker(\delta^{-1}_{-1}) = 0$}}, the algorithm stops at degree $-1$, giving:
$$\widetilde{X} = X_{0} \oplus \langle M^{*}_{a}\rangle_{-1}, \quad a=1, \dots, 4.$$ 
Suppose that {\small{$S_{0} \in \Pol_{\mathbb{R}}(M_{a})\setminus\Pol_{\mathbb{R}}(M_{4})$}}. Then, with an explicit computation one can verify that:
$$\Ker(\delta^{-1}_{-1}) = \langle \beta_{i},\gamma_{p}\rangle, \quad \mbox{ for } \quad \beta_{i}= \sum_{j, k} \epsilon_{ijk}M_{j}M_{k}^{*}, \quad \gamma_{p} = BM_{p}^{*}-M_{p}AM_{4}^{*}$$
where {\small{$\epsilon_{ijk}$}} is a totally antisymmetric tensor with {\small{$\epsilon_{123}= 1$}}, $i,p=1,2,3$, and {\small{$A$}}, {\small{$B \in\Pol(M_a)$}} coprime such that 
$$ \partial_{i}S_{0}=M_{i}AD  \quad \mbox{ and } \quad \partial_{4}(S_{0})=BD, \quad  \mbox{ with } \quad D:=GCD(\partial_{1}S_{0}, \dots, \partial_{4}S_{0}).$$
Analogous computations lead to the following results:
\begin{equation}
\label{H1}
\Imag(\delta^{-1}_{-2}) = \langle AD\beta_{i}, D\gamma_{p}\rangle  \quad \quad \mbox{ and so } \quad \quad H^{-1}(A)= \frac{\langle \beta_{i}, \gamma_{p}\rangle}{\langle AD\beta_{i}, D\gamma_{p}\rangle} 
\end{equation}
for $i, p=1, 2, 3$. Hence, two possibilities appear, depending whether the partial derivatives {\small{$\partial_{i}S_{0}$}} are coprime or not. \\
\\
\noindent
 {\em Case 2}: Suppose that the partial derivatives {\small{$\partial_{i} S_{0}$}} are all coprime, that is, $D=1$. Then, the cohomology group we are interested in reduces to: 
$$H^{-1}(A)=  \frac{\langle \beta_{1}, \beta_{2}, \beta_{3}\rangle}{\langle A \beta_{1}, A \beta_{2}, A\beta_{3}\rangle}\ .$$
Therefore, as prescribed by the algorithm, we introduce three real variables {\small{$C_{i}^{*}$}} with degree $-2$, which determine a graded algebra {\small{$A^{-2}$}} defined to be:
 \begin{equation}
A^{-2}:=\Pol_{\mathbb{R}}(M_a)\langle M_a^*, C^{*}_{i}\rangle\ , 
\end{equation}
for $a=1, \dots 4$, $i= 1,2 ,3$. Concerning the action of the coboundary operator {\small{$\delta^{-2}$}} on this graded algebra, it is uniquely determined by imposing linearity, being graded derivative and homogeneous on {\small{$\Pol_{\mathbb{R}}(M_a)$}} together with fixing its action on the $-2$-degree generators {\small{$C^{*}_{i}$}} to be {\small{$\delta^{-2}(C^{*}_{i})=  \beta_{i}$}}, for $i=1, \dots, 3$. Continuing with the algorithm, we see that an extra variable has to be introduced in degree $-3$. Indeed, there is only one generator $\xi$ of type $\beta$ for the cohomology group {\small{$H^{-2}(A^{-2})$}}, which has the following form:
 $$\xi=M_1C^*_{1}+M_2C^*_{2}+M_3C^*_{3}\ .$$
Hence {\small{$A^{-3}$}} is defined to be the following extension of the graded algebra {\small{$A^{-2}$}}: 
 \begin{equation}
A^{-3}:=\Pol(M_a)\langle M_a^*,  C^{*}_{i}, E^*\rangle\ ,
\end{equation}
with {\small{$E^{*}$}} a Grassmannian variable of degree $-3$. Moreover, {\small{$A^{-3}$}} can be endowed with a differential graded algebra structure, where the coboundary operator {\small{$\delta^{-3}$}} is defined to be the extension of {\small{$\delta^{-2}$}} to {\small{$A^{-3}$}} satisfying {\small{$\delta^{-3}(E^{*}) = \xi.$}} Because in degree $-3$ we have introduced only one independent variable it is straightforward to conclude that the algorithm stops at this step. This concludes the proof of the theorem in the coprime case. \\
\\
\noindent
{\em Case 3}: Under the hypothesis that the partial derivatives {\small{$\partial_{i}S_{0}$}} are not coprime, the cohomology group {\small{$H^{-1}(A)$}} in (\ref{H1}) has six independent generators. Then, six real independent variables {\small{$C^{*}_{j}$}} of degree $-2$ have to be introduced. Thus we define the algebra {\small{$A^{-2}$}} as the following extension: 
$$A^{-2}:=\Pol_{\mathbb{R}}(M_a)\langle M_a^*, C^{*}_{j}\rangle\ , $$
for $j=1, \dots, 6$ and with the coboundary operator $\delta$ on {\small{$A^{-1}$}} extended to a coboundary operator on {\small{$A^{-2}$}} by imposing {\small{$\delta^{-2}(C^{*}_{i}) = \beta_{i},$}} for $ i =1, \dots, 6$, to be the action on the generators of degree $-2$. Continuing with the algorithm, the variables of degree $-3$ that have to be introduced are determined by the linear relations with coefficients in {\small{$\Pol_{\mathbb{R}}(M_a)$}} existing among the elements {\small{$\delta(C_{1}^{*}), \dots, \delta(C_{6}^{*})$}}. With an explicit computation, we establish the following type-$\beta$ generators:
$$\alpha_{1}= \sum_{i} M_iC^*_{i}, \quad   \quad \alpha_{i+1}= - BC_{i}^{*} - \sum_{j,k} \epsilon_{ijk} M_{j}C_{k+3}^{*}, \quad \quad i, j, k=1, 2, 3. 
$$
Corresponding to these four generators we introduce four Grassmannian variables {\small{$E^{*}_{l}$}} obtaining that:
$$A^{-3}:=\Pol(M_a)\langle M_a^*,  C^{*}_{j}, E^{*}_{l}\rangle\ , $$
for $a, l= 1, \dots, 4$, $j=1, \dots, 6$. As usual, we extend {\small{$\delta^{-2}$}} to a coboundary operator on {\small{$A^{-3}$}} by imposing that {\small{$\delta^{-3}(E^{*}_{l}) = \alpha_{l},$}} for $l=1, \dots, 4$. Finally, we find that there is only one generator $\xi$ of type $\beta$ in degree $-3$, which is
$$\xi= BE^{*}_{1} + \sum_{i=1}^{3}M_iE^{*}_{i+1}\ .$$
Thus an extra real variable {\small{$K^{*}$}} of degree $-4$ has to be added, determining
$$A^{-4}:=\Pol(M_a)\langle M_a^*, C^{*}_{j}, E^{*}_{l}, K^*\rangle\, 
$$
for $a, l=1, \dots 4$, $ j=1, \dots, 6$ and with, as coboundary operator, the extension of {\small{$\delta^{-3}$}} via the requirement that {\small{$\delta^{-4}(K^{*}):=\xi$}}. Having found only one generator in degree $-4$, the construction automatically stops at this stage. Hence, in the non-coprime case, the minimally-extended configuration space is
\begin{center}
{\small{$\begin{array}{ll}
\widetilde{X} = & \langle K^* \rangle_{-4} \oplus \langle E^{*}_{1}, \dots, E^*_4\rangle_{-3} \oplus \langle C^{*}_{1}, \cdots, C^{*}_{6}\rangle_{-2} \oplus \langle M^*_{1}, \dots, M^*_{4}\rangle_{-1}\\
[.8ex]
&  \oplus X_{0} \oplus \langle C_{1}, \cdots, C_{6} \rangle_{1} \oplus \langle E_{1}, \dots, E_4\rangle_{2} \oplus \langle K\rangle_{3}.
\end{array}$}}
\end{center}
\vspace{-3mm}
\end{proof}

To complete the construction of an extended theory {\small{$(\widetilde{X}, \widetilde{S})$}} associated  to the initial theory {\small{$(X_{0}, S_{0})$}} we still have to determine the extended action {\small{$\widetilde{S}$}}, that is, a functional {\small{$\widetilde{S}: \widetilde{X} \rightarrow \mathbb{R}$}} that solves the classical master equation on {\small{$\widetilde{X}$}}. Because we have found different extended configuration spaces depending on the properties of the initial action {\small{$S_{0}$}}, these would determine different extended actions. We concentrate on the generic case, which is the one described in point $(2)$ of Theorem \ref{Theorem: extended theory for the U(2) model}. In this setting, we have the following result.

\begin{theorem}
\label{BV variety model}
 Let {\small{$(X_{0}, S_{0})$}} be a gauge theory with configuration space {\small{$X_{0} \cong A^{4}_{\mathbb{R}}$}} and action functional {\small{$S_{0}$}} as in (\ref{S_0 generale}) such that {\small{$D:=GCD(\partial_{i}S_{0}) =1$}}. Then, given an extended configuration space {\small{$\widetilde{X}$}} as in Theorem \ref{Theorem: extended theory for the U(2) model}, the most general solution of the classical master equation on {\small{$\widetilde{X}$}} that is linear in the anti-fields, of at most degree $2$ in the ghost fields and with coefficients in {\small{$\Pol_{\mathbb{R}}(M_{a})$}} is the following one:
{\small{\begin{multline*}
\widetilde{S}= S_{0} + \sum_{i, j, k} \epsilon_{ijk}\alpha_{k}M_{i}^{*}M_{j}C_{k} + \sum_{i, j,k} C_{i}^{*}\big[\tfrac{\alpha_{j}\alpha_{k}}{2\alpha_{i}} (\beta\alpha_{i}M_{i}E + \epsilon_{ijk}C_{j}C_{k}) \\
+ M_{i}T \big( \sum_{a, b,c} \epsilon_{abc}\tfrac{\alpha_{b}\alpha_{c}}{2\alpha_{i}}M_{a}C_{b}C_{c}\big)\big]
\end{multline*}}}
where $\alpha_{i}, \beta\in {\small{\mathbb{R}\backslash \left\lbrace 0 \right\rbrace}}$, {\small{$T \in \Pol_{\mathbb{R}}(M_a)$}}, and $\epsilon_{ijk}$ ($\epsilon_{abc}$) is the totally anti-symmetric tensor in three (different) indices $i, j, k \in \{ 1, 2, 3 \}$ ($a, b, c \in \{ 1, 2,3\}$) with $\epsilon_{123}=1$.
\end{theorem}

\begin{proof}
We start considering the approximation of the extended action that is linear in the positively-graded generators, that is, 
$$\begin{array}{ll}
\widetilde{S}_{\leqslant 1} & = S_{0} + \sum_{i, j, k} \epsilon_{ijk}\alpha_{k}M_{i}^{*}M_{j}C_{k} + \sum_{i, j,k} C_{i}^{*}\big[\tfrac{\alpha_{j}\alpha_{k}}{2\alpha_{i}} \beta\alpha_{i}M_{i}E \big]\ .
\end{array}
$$
 We explicitly compute the quantity {\small{$\big\{ \widetilde{S}_{\leqslant 1}, \widetilde{S}_{\leqslant 1}\big\}$}}, obtaining the following expression:
\begin{equation}
\big\{ \widetilde{S}_{\leqslant 1}, \widetilde{S}_{\leqslant 1}\big\} = 2\big[ \sum_{i,j} \alpha_{i}\alpha_{j}M_{i}M^{*}_{j}C_{i}C_{j} + \sum_{i, j, k} \beta \epsilon_{ijk} \alpha_{j} \alpha_{k}^{2} M_{j}C_{k}C_{i}^{*}E \big] ,
\label{compensare}
\end{equation}
with $i, j, k=1, 2, 3$. Since {\small{$\big\{ \widetilde{S}_{\leqslant 1}, \widetilde{S}_{\leqslant 1}\big\}\neq 0$}}, to construct the approximated action {\small{$\widetilde{S}_{\leqslant 2}$}}, we need to introduce a generic element {\small{$\nu \in  I^{\geqslant 2}_{\widetilde{X}}\cap F^{2}\mathcal{O}_{\widetilde{X}}$}}, which is determined by imposing
\begin{equation}
\label{condizione per determinare nu}
2 (\delta \otimes Id) \nu + \big\{ \widetilde{S}_{\leqslant 1}, \widetilde{S}_{\leqslant 1}\big\} \equiv 0 \quad  \mbox{  mod  } F^{3}\mathcal{O}_{\widetilde{X}}\ .
\end{equation}
Because in (\ref{compensare}) do not appear terms depending on any {\small{$\partial_{i}S_{0}$}}, we could restrict to consider an element {\small{$\nu$}} not depending on the anti-fields {\small{$M_{a}^{*}$}}. Moreover, since condition (\ref{condizione per determinare nu}) is taken modulo {\small{$F^{3}\mathcal{O}_{\widetilde{X}}$}}, we could start considering a {\small{$\nu$}} which is precisely of positive degree $2$, because terms in {\small{$\nu$}} of higher positive degree would contribute quantities which automatically belongs to {\small{$F^{3}\mathcal{O}_{\widetilde{X}}$}}. Thus we consider an element {\small{$\nu$}} whose generic form is the following one:
$$\nu= \sum_{i<j} g_{k}^{ij} C_{k}^{*}C_{i}C_{j}$$
for $i,j,k=1, 2, 3$, and {\small{$g_{k}^{ij} \in \Pol_{\mathbb{R}}(M_{a})$}}.  Hence, to satisfy \eqref{condizione per determinare nu}, the polynomials {\small{$g_{k}^{ij}$}} should verify the following equalities:
\begin{equation}
\alpha_{[i}g_{i}^{ij}M_{j]}=0, \quad \quad \alpha_{i}\alpha_{j}M_{i} + \epsilon_{ijk}\alpha_{[i}g_{i}^{ij}M_{k]} =0,
\end{equation}
where $i, j, k=1, 2, 3$ are different, the polynomials {\small{$g_{k}^{ij}$}} are antisymmetric in the top indices and the bracket enclosing the indices denotes the anti-symmetrization in the indices themselves, that is:
$$\alpha_{[i}g_{i}^{ij}M_{j]} := \alpha_{i}g_{i}^{ij}M_{j} - \alpha_{j}g_{j}^{ij}M_{i}.$$
Hence: 
$$
g_{i}^{ij} = \epsilon_{ijk}\alpha_{j}M_{i}P_{k}, \quad g_{j}^{ij} = \epsilon_{ijk} \alpha_{i}M_{j}P_{k}, \quad g_{k}^{ij} =\epsilon_{ijk}\frac{\alpha_{i} \alpha_{j}}{\alpha_{k}}[1 + M_{k}P_{k}], 
$$
with $i, j, k=1, 2, 3$, $i<j$, {\small{$P_{i}\in \Pol_{\mathbb{R}}(M_{a})$}}. Therefore, an approximated action that solves the classical master equation up to terms of positive degree $2$ is {\small{$\widetilde{S}_{\leqslant 2} := \widetilde{S}_{\leqslant 1} + \nu$}}, where the polynomials {\small{$g^{ij}_{k}$}} are the ones just determine. To proceed with the algorithm, we compute the quantity {\small{$\big\{ \widetilde{S}_{\leqslant 2}, \widetilde{S}_{\leqslant 2}\big\}$}} and we verify that the freedom in choosing the polynomials  $P_{i}$ can be used to convert the approximated solution {\small{$\widetilde{S}_{\leqslant 2}$}} in an exact solution to the classical master equation by imposing
$$P_{i}= M_{i}T,$$
where {\small{$T\in \Pol_{\mathbb{R}}(M_{a})$}} is still a free parameter. Hence, with this choice of polynomials {\small{$P_{i}$}}, the functional {\small{$\widetilde{S}_{\leqslant 2}$}} satisfies all the required conditions and the claimed statement immediately follows. 
\end{proof}

\begin{oss}
The relation between the initial gauge theory {\small{$(X_{0}, S_{0})$}}  and the minimal extended theory {\small{$(\widetilde{X}, \widetilde{S})$}} or, in other words, the role played by the choice of a basis for {\small{$X_{0}$}} still deserves further investigations. Indeed, while it is convincing that the BV construction detects the structure of the gauge group, as we have seen in the proof of Theorem \ref{Theorem: extended theory for the U(2) model}, the procedure to relate different extended theories corresponding to different choices for the basis of {\small{$X_{0}$}} still has to be determined and it is left for future researches.
\end{oss}

With the result proved in Theorem \ref{BV variety model} we have constructed an extended theory corresponding to a {\small{$U(2)$}}-matrix model. Not only is this result of interest given the current lack of examples of a BV-extended theory, but also it can be viewed as the first step towards the construction of a (gauge-fixed) BRST-cohomology complex for this {\small{$U(2)$}}-matrix model. The analysis of this cohomology is beyond the scope of this article and we address it in \cite{articolo_cohomology}.

\bibliographystyle{plain}

\begin{thebibliography}{99}

\bibitem{AKSZ}
M. Alexandrov and M. Kontsevich and A. Schwarz and O.Zabronsky
\newblock The geometry of the master equation and topological quantum field theory.
\newblock {\em Int. J. Mod. Phys.}, A 12, (1997), 1405-1430.

\bibitem{BBH}
G. Barnich and F. Brandt and M. Henneaux.
\newblock Local {BRST} cohomology in gauge theories.
\newblock {\em Phys. Rep.}, 338, (2000), 439-569.

\bibitem{BBH2}
G. Barnich and F. Brandt and M. Henneaux.
\newblock Local {BRST} cohomology in the antifield formalism {I}. {G}eneral theorems.
\newblock {\em Commun. Math. Phys.}, 174, (1995), 57-92.

\bibitem{BV1}
I.A. Batalin and G.A. Vilkovisky.
\newblock Gauge algebra and quantization.
\newblock {\em Phys. Lett.}, B 102, (1981), 27-31.

\bibitem{BV2}
I.A. Batalin and G.A. Vilkovisky.
\newblock Quantization of gauge theories with linearly dependent generators.
\newblock {\em Phys. Rev.}, D28 (1983), 2567-2582, Erratum, D30 (1984), 508.

\bibitem{BRS}
C.M. Becchi, A.~Rouet, and R.~Stora.
\newblock Renormalization of gauge theories.
\newblock {\em Ann. Phys.}, 98, 2 (1976), 287-321.

\bibitem{BRS3}
C.M. Becchi, A.~Rouet, and R.~Stora.
\newblock Renormalization of the abelian {H}iggs-{K}ibble model.
\newblock {\em Commun. Math. Phys.}, 42, (1975), 127-162.

\bibitem{beringgrosse}
  K.~Bering and H.~Grosse.
  \newblock On Batalin-Vilkovisky Formalism of Non-Commutative Field Theories.
  \newblock {\em Eur.\ Phys.\ J.\ }C {\bf 68} (2010) 313.

\bibitem{Costello}
K. Costello and O. Gwilliam.
\newblock {\em Factorization algebras in quantum field theory}.
\newblock Cambridge University Press, (2014).

\bibitem{Faddeev-Popov}
L.D. Faddeev and V.N. Popov.
\newblock Feynman diagrams for the {Y}ang-{M}ills field.
\newblock {\em Phys. Lett.}, B 25, (1967), 29-30.

\bibitem{felder}
G.~Felder and D.~Kazhdan.
\newblock The classical master equation.
\newblock In {\em Perspectives in Representation Theory, Contemporary
  Mathematics, \em{Eds P. Etingof, M. Khovanov, and A. Savage, (2014)}}.

\bibitem{Feynman}
R.P. Feynman and A.R. Hibbs.
\newblock Quantum mechanics and path integrals.
\newblock {\em McGraw-Hill, New York}, (1965).

\bibitem{Fior}
D.~Fiorenza.
\newblock An introduction to the {B}atalin-{V}ilkovisky formalism.
\newblock {\em Comptes Rendus des Rencontres Mathématiques de Glanon}, (2003).

\bibitem{Fuster}
A. Fuster and M. Henneaux and A. Maas.
\newblock {BRST} quantization: a short review.
\newblock {\em Int. J. Geom. Meth. Mod. Phys.}, 2, (2005), 939-964.

\bibitem{GPS}
J.~Gomis, J.~París, and S.~Samuel.
\newblock Antibracket, antifields and gauge-theory quantization.
\newblock {\em Phys. Rep.}, 259, (1995), 1-145.

\bibitem{Henneaux2}
M. Henneaux.
\newblock Lectures on the antifield - {BRST} formalism for gauge theories.
\newblock {\em Nucl. Phys. Proc. Suppl.}, 18 A, (1990), 47-106.

\bibitem{articolo_cohomology}
R.A. Iseppi.
\newblock The {BRST} cohomology and a generalized {L}ie algebra cohomology: analysis of a matrix model.
\newblock {\em In preparation}.

\bibitem{Schw}
A.~Schwarz.
\newblock Geometry of {B}atalin-{V}ilkovisky quantization.
\newblock {\em Commun. Math. Phys. }, 155, (1993), 249-260.

\bibitem{Tate}
J. Tate.
\newblock Homology of {N}oetherian rings and local rings.
\newblock {\em Illinois Journal of Mathematics}, 1, (1957), 14-27.

\bibitem{T}
I.V. Tyutin.
\newblock Gauge invariance in field theory and statistical physics in operator
  formalism.
\newblock {\em Preprint of P.N. Lebedev Physical Institute}, 39, (1975).

\bibitem{Zinn-Justin}
J. Zinn-Justin.
\newblock Renormalization of gauge theories.
\newblock {\em in Trends in Elementary Particle Theory}, Eds: H.Rollnik and K. Dietz, Lecture Notes in Physics, Vol 37, Springer-Verlag, Berlin, (1975).

\end{thebibliography}

\end{document}